\documentclass[11pt]{article}
\usepackage{waingarten} 
\author{Moses Charikar\thanks{email: \texttt{moses@cs.stanford.edu}. Moses Charikar is supported by a Simons Investigator Award.} \\ Stanford University \and Erik Waingarten\thanks{email: \texttt{ewaingar@seas.upenn.edu}. Part of this work was done while Erik Waingarten was a postdoc at Stanford University, supported by an NSF postdoctoral fellowship and by Moses Charikar's Simons Investigator Award.} \\ University of Pennsylvania}
\title{The Johnson-Lindenstrauss Lemma for Clustering and Subspace Approximation: From Coresets to Dimension Reduction}

\usepackage{hyperref}
\hypersetup{
    colorlinks=true,
    linkcolor=blue,
    filecolor=magenta,      
    urlcolor=cyan,
    pdftitle={Overleaf Example},
    pdfpagemode=FullScreen,
    }
\hypersetup{
    colorlinks,
    linkcolor=black,
    filecolor=blue,      
    urlcolor=blue,
    citecolor=blue,
    linktoc=all
}

\begin{document}         
\maketitle
\thispagestyle{empty}

\begin{abstract}
We study the effect of Johnson-Lindenstrauss transforms in various projective clustering problems, generalizing recent results which only applied to center-based clustering~\cite{MMR19}. We ask the general question: for a Euclidean optimization problem and an accuracy parameter $\eps \in (0, 1)$, what is the smallest target dimension $t \in \N$ such that a Johnson-Lindenstrauss transform $\bPi \colon \R^d \to \R^t$ preserves the cost of the optimal solution up to a $(1+\eps)$-factor. We give a new technique which uses coreset constructions to analyze the effect of the Johnson-Lindenstrauss transform. 
Our technique, in addition applying to center-based clustering, improves on (or is the first to address) other Euclidean optimization problems, including:
\begin{itemize}
\item For $(k,z)$-subspace approximation: we show that $t = \tilde{O}(zk^2 / \eps^3)$ suffices, whereas the prior best bound, of $O(k/\eps^2)$, only applied to the case $z = 2$ \cite{CEMMP15}.
\item For $(k,z)$-flat approximation: we show $t = \tilde{O}(zk^2/\eps^3)$ suffices, completely removing the dependence on $n$ from the prior bound $\tilde{O}(zk^2 \log n/\eps^3)$ of~\cite{KR15}.
\item For $(k,z)$-line approximation: we show $t = O((k \log \log n + z + \log(1/\eps)) / \eps^3)$ suffices, and ours is the first to give any dimension reduction result.
\end{itemize}
\end{abstract}

\newpage
\thispagestyle{empty}
\tableofcontents

\thispagestyle{empty}

\newpage 
\setcounter{page}{1}


\section{Introduction}

\newcommand{\cost}{\mathrm{cost}}

The Johnson-Lindenstrauss lemma \cite{JL84} concerns dimensionality reduction for high-dimensional Euclidean spaces. It states that, for any set of $n$ points $x_1,\dots, x_n$ in $\R^d$ and any $\eps \in (0,1)$, there exists a map $\Pi \colon \R^d \to  \R^t$, with $t = O(\log n / \eps^2)$ such that, for any $i, j \in [n]$,
\begin{align}
 \dfrac{1}{1+\eps} \cdot \|x_i - x_j \|_2 \leq \|\Pi(x_i) - \Pi(x_j)\|_2 \leq (1+\eps) \cdot \| x_i - x_j \|_2. \label{eq:preserve-dist}
 \end{align}
From a computational perspective, the lemma has been extremely influential in designing algorithms for high-dimensional geometric problems, partly because proofs show that a \emph{random} linear map, oblivious to the data, suffices. Proofs specify a distribution $\calJ_{d,t}$ supported on linear maps $\R^d \to \R^t$ which is independent of $x_1,\dots, x_n$ (for example, given by a $t \times d$ matrix of i.i.d $\calN(0, 1/t)$ entries \cite{IM98, DG03}), and show that a draw $\bPi \sim \calJ_{d,t}$ satisfies (\ref{eq:preserve-dist}) with probability at least $0.9$.

In this paper, we study the Johnson-Lindenstrauss transforms for \emph{projective clustering} problems, generalizing a recent line-of-work which gave (surprising) dimension reduction results for center-based clustering~\cite{BZD10, CEMMP15, BBCGS19, MMR19}. The goal is to reduce the dimensionality of the input of a (more general) projective clustering problem (from $d$ to $t$ with $t \ll d$) without affecting the cost of the optimal solution significantly. We map $d$-dimensional points to $t$-dimensional points (via a random linear map $\bPi$) such that the optimal cost in $t$-dimensions is within a $(1+\eps)$-factor of the optimal cost in the original $d$-dimensional space. We study this for a variety of problems, each of which is specified by a set of candidate solutions $\calC_d$ and a cost function. By varying the family of candidate solutions $\calC_d$ and the cost functions considered, one obtains center-based clustering problems (like $k$-means and $k$-median), as well as subspace approximation problems (like principal components analysis), and beyond (like clustering with subspaces). The key question we address here is:
\begin{quote}
\textbf{Main Question}: For a projective clustering problem, how small can $t$ be as a function of $n$ (the dataset size) and $\eps$ (the accuracy parameter), such that the cost of the optimization is preserved up to $(1\pm \eps)$-factor with probability at least $0.9$ over $\bPi \sim \calJ_{d,t}$?
\end{quote}

Our results fit into a line of prior work on the power of Johnson-Lindenstrauss maps beyond the original discovery of \cite{JL84}. These have been investigated before for various problems and in various contexts, including nearest neighbor search~\cite{IM98, HIM12, AIR18}, numerical linear algebra~\cite{S06, M11, W14}, prioritized and terminal embeddings \cite{EFN17, MMMR18, NN19, CN21}, and clustering and facility location problems~\cite{BZD10, CEMMP15, KR15, BBCGS19, MMR19, NSIZ21, ISZ21}. 


\subsection{Our Contribution}

In the $(k, j)$-projective clustering problem with $\ell_z$ cost (\cite{HV02, AP03, EV05, DRVW06,VX12, KR15, FSS20, TWZBF22}), the goal is to cluster a dataset $X = \{ x_1,\dots, x_n \} \subset \R^d$, where each cluster is approximated by an affine $j$-dimensional subspace. Namely, we define an objective function for $(k, j)$-projective clustering problems with $\ell_z$ cost on a dataset $X$, which aims to minimize
\[ \min_{c \in \calC_d} \cost_z(X,c), \]
where the ``candidate solutions'' $\calC_{d}$ consist of all $k$-tuples of $j$-dimensional affine subspaces. 
The cost function $\cost_z(X, \cdot)$ maps each candidate solution $c \in \calC_d$ to a cost in $\R_{\geq 0}$ given by the $\ell_z$-norm of the vector of distances between each dataset point $x \in X$ to its nearest point on one of the $k$ subspaces. Intuitively, each point of the dataset $x_i \in X$ ``pays'' for the Euclidean distance to the nearest subspace in $c \in \calC_d$, and the total cost is the $\ell_z$-norm of the $n$ payments (one for each point).

There has been significant prior work which showed surprising results for the special case of $(k,z)$-clustering (like $k$-means and $k$-median, which corresponds to $(k, 0)$-projective clustering with $\ell_z$-cost) as well as for low-rank approximation (which corresponds to $(1, k)$-projective clustering for non-affine subspaces with $\ell_2$-cost)~\cite{BZD10, CEMMP15, BBCGS19, MMR19}. It is important to note that the techniques in prior works are \emph{specifically tailored} to the Euclidean optimization problem at hand. For example, the results of~\cite{MMR19}, which apply for $(k,0)$-clustering with $\ell_z$-norm, rely on using center points as the approximation and do not generalize to affine subspaces beyond dimension $0$. The other result of~\cite{CEMMP15} for low-rank approximation uses the specific algebraic properties of the $\ell_2$-norm (which characterize the optimal low-rank approximation). These prior works carry a strong conceptual message: for $(k,z)$-clustering and low-rank approximation, even though many pairwise distances among dataset points become highly distorted (since one projects to $t \ll \log n$ dimensions), the cost of the optimization (which aggregates distances) need not be significantly distorted. 

\textbf{Our Results.} We show that $(k, 0)$-clustering with $\ell_z$-norm and low-rank approximation are not isolated incidents, but rather, part of a more general phenomenon. Our main conceptual contribution is the following: we use algorithms for constructing \emph{coresets} (via the ``sensitivity sampling'' framework of \cite{FL11}) to obtain bounds on dimension reduction. Then, the specific bounds that we obtain for the various problems depend on the sizes of the coresets that the algorithms can produce. We can instantiate our framework to new upper bounds for the following problems:
\begin{itemize}
\item \textbf{$(k, z)$-Subspace Approximation}. This problem is a (restricted) $(1, k)$-projective clustering problems with $\ell_z$-cost. We seek to minimize over a $k$-dimensional subspace $S$ of $\R^d$ the $\ell_z$-norm of the $n$-dimensional vector where the coordinate $i \in [n]$ encodes the distance between $x_i$ and the closest point in $S$.\footnote{This is a restricted version of projective clustering because subspaces are not affine and required to go through the origin.} 
\item \textbf{$(k,z)$-Flat Approximation}. This problem is exactly the $(1.k)$-projective clustering problem with $\ell_z$ cost. It is similar to $(k,z)$-subspace approximation, except we optimize over all affine subspaces. 
\item \textbf{$(k,z)$-Line Approximation}. This problem corresponds to $(k, 1)$-projective clustering with $\ell_z$-cost. The optimization minimizes, over an arbitrary set $L$ of $k$ of $1$-dimensional affine subspaces $l_1,\dots, l_k$ (i.e., $k$ lines in $\R^d$), the $\ell_z$-norm of the $n$-dimensional vector where the coordinate $i \in [n]$ encodes the distance between $x_i$ and the closest point on any line in $L$. 
\end{itemize}
Concretely, our quantitative results are summarized by the following theorem.
\begin{theorem}[Main Result---Informal]\label{thm:main}
Let $X = \{ x_1,\dots, x_n \} \subset \R^d$ be any dataset, and $\calJ_{d, t}$ denote a $t \times d$ matrix of i.i.d entries from $\calN(0, 1/t)$. Let $\calC_d$ and $\calC_t$ be candidate solutions for a projective clustering problem in $\R^d$ and $\R^t$, respectively. For any $\eps \in (0, 1)$, we have
\begin{align}
\Prx_{\bPi \sim \calJ_{d, t}}\left[ \left| \dfrac{\min_{c' \in \calC_{t}} \cost(\bPi(X), c')}{\min_{c \in \calC_d} \cost(X, c)} - 1\right| \leq \eps  \right] \geq 0.9 \label{eq:comparison}
\end{align}
whenever:
\begin{itemize}
\item \emph{$(k, z)$-subspace and $(k,z)$-flat approximation:} $\calC_d$ and $\calC_t$ are all $k$-dimensional subspaces of $\R^d$ and $\R^t$, respectively; the cost measures the $\ell_z$-norm of distances between points to the subspace; and, $t = \tilde{O}(zk^2/\eps^3)$. Similarly, the same bound on $t$ holds for $\calC_d$ and $\calC_t$ varying over all affine $k$-dimensional subspaces.
\item \emph{$(k, z)$-line approximation:} $\calC_d$ and $\calC_t$ are all $k$-tuples of lines in $\R^d$ and $\R^t$, respectively; the cost measures the $\ell_z$-norm of distances between points and the nearest line; and, $t = O( (k \log\log n + z + \log(1/\eps)) / \eps^3)$. 
\end{itemize}
\end{theorem}
In all cases, the bound that we obtain is directly related to the size of the best coresets from the sensitivity sampling framework, and all of our proofs follow the same format.\footnote{The reason we did not generalize the results to $(k, j)$-projective clustering with $\ell_z$-cost with $j > 1$ is that these problems do not admit small coresets~\cite{H04}. Researchers have studied ``integer $(k, j)$-projective clustering'' where one restricts the input points to have bounded integer coordinates, where small coresets do exists~\cite{EV05}. However, using this approach for dimension reduction would incur additional additive errors, so we have chosen not to pursue this route.} Our proofs are not entirely black-box applications of coresets (since we use the specific instantiations of the sensitivity sampling framework), but we believe that any improvement on the size of the best coresets will likely lead to quantitative improvements on the dimension reduction bounds. However, improving our current bounds by either better coresets or via a different argument altogether (for example, improving on the cubic dependence on $\eps$) seems to require significantly new ideas.

\begin{figure}[H]
\begin{center}
\begin{tabular}{ |c||c|c|c|  } 
 \hline
 \multicolumn{3}{|c|}{Target Dimension for Johnson-Lindenstrauss Transforms} \\
 \hline
Problem & New Result & Prior Best \\
 \hline
$(k,z)$-subspace & $\tilde{O}(zk^2 / \eps^3)$  & $O(k/\eps^2)$ (when $z = 2$) \small{\cite{CEMMP15}} \\
 $(k,z)$-flat & $\tilde{O}(zk^2/\eps^3)$ & $\tilde{O}(zk^2 \log n/\eps^3)$ \small{\cite{KR15}} \\
 $(k,z)$-line    &$O((k\log\log n + z + \log(1/\eps))/\eps^3)$ & None \\
\hline
\end{tabular}
\end{center}
\caption{Comparison to Prior Bounds
}
\end{figure}

\textbf{A Subtlety in ``For-All'' versus ``Optimal'' Guarantees.} So far, the main question (and the results we obtain) focus on applying the Johnson-Lindenstrauss transform and preserving the optimal cost, i.e., that the \emph{minimizing solution} in the original and the dimension reduced space have approximately the same cost. A stronger guarantee which one may hope for, a so-called ``for-all'' guarantee, asks that after applying the Johnson-Lindenstrauss transform, every solution has its cost approximately preserved before and after dimension reduction. We do not achieve ``for all'' guarantees, like those appearing in~\cite{MMR19}. However, we emphasize that various subtleties arise in what is meant by ``a solution,'' as the prior work on dimension reduction and coresets refer to \emph{different notions} (even though they agree at the optimum). 

Consider, for example, the $1$-medoid problem, a constrained version of the $1$-means problem. The $1$-medoid cost of a dataset $X$ is the minimum over centers $c$ chosen from the dataset $X$, of the sum of squares of distances from each dataset point $x$ to $c$. The subtlety is the following: one can apply a Johnson-Lindenstrauss transform to $t = O(\log(1/\eps)/\eps^2)$ dimensions and preserve the $1$-means cost, and one may hope that a ``for-all'' guarantee would also preserve the $1$-medoid cost. Somewhat surprisingly, we show that it does not. 
\begin{theorem}[Johnson-Lindenstrauss for Medoid---Informal (see Theorem~\ref{thm:medoid})]\label{thm:medoid-inf}
For any $t = o(\log n)$, applying a Johnson-Lindenstrauss transform to dimension $t$ decreases the cost of the $1$-medoid problem by a factor approaching $2$. 
\end{theorem}

\textbf{Not all ``for-all'' guarantees are the same.} A ``for-all'' guarantee comes with (an implicit) representation of a candidate ``solution,'' and different choices of representations yield different guarantees. The above theorem does not contradict the ``for-all'' guarantee of~\cite{MMR19} because there, a candidate solution for $(k,z)$-clustering refers to a partition of $X$ into $k$ parts and not a set of centers. Often in the coreset literature, a ``solution'' refers to a set of centers and not arbitrary partitions. For $k = 1$, there are many possible centers but only one partition, and it is important (as per Theorem~\ref{thm:medoid-inf}) that a potential ``for all'' guarantee considers partitions.

For $(k, z)$-subspace and -flat approximation, in a natural representation of the solutions, the same issue arises. Consider the $1$-column subset selection problem, a constrained version of the $(1,2)$-subspace approximation problem, where subspaces must be in the span of the dataset points. The $1$-column subset selection cost of a dataset is the minimum over $1$-dimensional subspaces spanned by a dataset point of $X$, of the sum of squares of distances from each dataset point $x$ to the projection onto the subspace. Similarly to Theorem~\ref{thm:medoid-inf}, a Johnson-Lindenstrauss transform does not preserve the cost of $1$-column subset selection.
\begin{theorem}[Johnson-Lindenstrauss for Column Subset Selection---Informal (see Theorem~\ref{thm:1-css})]
For any $t = o(\log n)$, applying a Johnson-Lindenstrauss transform to dimension $t$ decreases the cost of the $1$-column subset selection problem by a factor approaching $3/2$.
\end{theorem}
The above theorem does not contradict the ``for-all'' guarantee of~\cite{CEMMP15} for similar reasons (which, in addition, crucially rely on having $z = 2$, and which we elaborate on in Appendix~\ref{sec:forall}). For $(k,z)$-line approximation, however, there is an interesting open problem: is it true that after applying a Johnson-Lindenstrauss transform to $t = \poly(k \log\log n/\eps)$, \emph{for all partitions} of $X$ into $k$ parts, the cost of \emph{optimally} approximating each part with a line has its cost preserved. 

\ignore{This section is meant for two things:
\begin{enumerate}
\item To help compare the guarantees of this work to that of prior works on $(k,z)$-clustering of \cite{MMR19} and $(k,2)$-subspace approximation \cite{CEMMP15}, expanding on Remark~\ref{rem:for-all}. In short, for $(k,z)$-clustering, the results of \cite{MMR19} are qualitatively stronger than the results obtained here. In $(k,2)$-subspace approximation, the ``for all'' guarantees  of \cite{CEMMP15} are for the qualitatively different problem of low-rank approximation. While the costs of low-rank approximation and $(k, 2)$-subspace approximation happen to agree at the optimum, the notion of a candidate solution is different.
\item To show that, for two related problems of ``medoid'' and ``column subset selection,'' one cannot apply the Johnson-Lindenstrauss transform to dimension $o(\log n)$ while preserving the cost. The medoid problem is a center-based clustering problem, and column subset selection problem is a subspace approximation problem. The instances we will construct for these problems are very symmetric, such that uniform sampling will give small coresets. The se give concrete examples ruling out a theorem which directly relates the size of coresets to the effect of the Johnson-Lindenstrauss transform.
\end{enumerate}

\paragraph{Center-Based Clustering} Consider the following (slight) modification to the center-based clustering problems known as the ``medoid'' problem.
\begin{definition}[$1$-medoid problem]
Let $X = \{ x_1,\dots, x_n \} \subset \R^d$ be any set of points. The $1$-medoid problem asks to optimize
\[ \min_{\substack{c \in X}} \sum_{x \in X} \| x - c\|_2^2. \]
\end{definition}
Notice the difference between $1$-medoid and $1$-mean: in $1$-medoid the center is restricted to be from within the set of points $X$, whereas in $1$-mean the center is arbitrary. Perhaps surprisingly, this modification has a dramatic effect on dimension reduction.
\begin{theorem}\label{thm:medoid}
For large enough $n, d \in \N$, there exists a set of points $X \subset \R^d$ (in particular, given by the $n$-basis vector $\{ e_1,\dots, e_n \} \subset \R^n$) such that, with high probability over the draw of $\bPi \sim \calJ_{d, t}$ where $t = o(\log n)$, 
\begin{align*}
\dfrac{\mathop{\min}_{c \in X} \sum_{x \in X} \| x - c\|_2^2} {\mathop{\min}_{c' \in \bPi(X)} \sum_{x \in X} \| \bPi(x) - c'\|_2^2}  \geq 2 - o(1).
\end{align*}
\end{theorem}

Theorem~\ref{thm:medoid}  gives very strong lower bound for dimension reduction for $k$-medoid, showing that decreasing the dimension to any $o(\log n)$ does not preserve (even the optimal) solutions within better-than factor $2$. This is in stark contrast to the results on center-based clustering, where the $1$-mean problem can preserve the solutions up to $(1\pm \eps)$-approximation without any dependence on $n$ or $d$. The proof itself is also very straight-forward: each $\bPi(e_i)$ is an independent Gaussian vector in $\R^t$, and if $t = o(\log n)$, with high probability, there exists an index $i \in [n]$ where $\| \bPi(e_i) \|_2^2 = o(1)$. In a similar vein, with high probability $\sum_{i=1}^n \| \bPi(e_i) \|_2^2 \leq (1+o(1)) n$. We take a union bound and set the center $c' = \bPi(e_i)$ for the index $i$ where $\| \bPi(e_i) \|_2^2 = o(1)$. By the pythagorean theorem, the cost of this $1$-medoid solution is at most $(1+o(1)) n$. On the other hand, every $1$-medoid solution in $X$ has cost $2(n-1)$. 

We emphasize that Theorem~\ref{thm:medoid} does not contradict~\cite{MMR19, BBCGS19}, even though it rules out that ``all candidate \emph{centers}'' are preserved. The reason is that the notion of ``candidate solution'' is different. Informally, \cite{MMR19} shows that for any dataset $X \subset \R^d$ of $n$ vectors and any $k \in \N$, $\eps > 0$, applying the Johnson-Lindenstrauss map $\bPi \sim \calJ_{d, t}$ with $t = \tilde{O}(\log(k/\eps) / \eps^2)$ satisfies the following guarantee: for all partitions of $X$ into $k$ sets, $(P_1, P_2, \dots, P_k)$, the following is true:
\begin{align*}
\sum_{\ell=1}^k \min_{c_{\ell}' \in \R^t} \sum_{x \in P_{\ell}} \| \bPi(x) - c_{\ell}'\|_2^2 \approx_{1\pm \eps} \sum_{\ell=1}^k \min_{c_{\ell} \in \R^d} \sum_{x \in P_{\ell}} \| x - c_{\ell}\|_2^2 .
\end{align*}
The ``for all'' quantifies over clusterings $(P_1,\dots, P_k)$ is different (as seen from the $1$-medoid example) from the ``for all'' over centers $c_1,\dots, c_k$.

\paragraph{Subspace Approximation} The same subtlety appears in subspace approximation. Here, we can consider the $1$-column subset selection problem, which at a high level, is the medoid version of subspace approximation. We want to approximate a set of points by their projections onto the subspace spanned by one of those points.
\begin{definition}[$1$-column subset selection]
Let $X = \{ x_1,\dots, x_n \} \subset \R^d$ be any set of points. The $1$-column subset selection problem asks to optimize
\begin{align*}
\min_{\substack{S = \mathrm{span}(\{ x_i \}) \\ x_i \in X}} \sum_{x \in X} \| x - \rho_{S}(x) \|_2^2
\end{align*}
\end{definition}

Again, notice the difference between $1$-column subset selection and $(k, 1)$-subspace approximation: the subspace $S$ is restricted to be in the span of a point from $X$. Given Theorem~\ref{thm:medoid}, it is not surprising that Johnson-Lindenstrauss cannot reduce the dimension of $1$-column subset selection to $o(\log n)$ without incurring high distortions.
\begin{theorem}\label{thm:1-css}
For large enough $n, d \in \N$, there exists a set of points $X \subset \R^d$ such that, with high probability over the draw of $\bPi \sim \calJ_{d,t}$ where $t = o(\log n)$, 
\begin{align*}
\dfrac{\mathop{\min}_{\substack{S = \mathrm{span}(x) \\ x \in X}} \sum_{x \in X} \| x - \rho_S(x)\|_2^2} {\mathop{\min}_{\substack{S' = \mathrm{span}(\bPi(x)) \\ x \in X}} \sum_{x \in X} \| \bPi(x) - \rho_{S'}(\bPi(x))\|_2^2}  \geq 3/2 - o(1).
\end{align*}
\end{theorem}

For $(k,z)$-subspace approximation, we want to approximate the entire dataset with a single subspace---there is no partition of the dataset. To define a natural ``for-all'' guarantee when $\calC_d$ and $\calC_t$ consist of the set of all $k$-dimensional subspaces of $\R^d$ and $\R^t$, one associates each $c \in \calC_d$ with $\tau(c) \in \calC_t$ and hope to prove a ``for-all'' analogue of (\ref{eq:comparison}) for a small value of $t$:
\[ \Prx_{\bPi \sim \calJ_{d, t}}\left[ \forall c \in \calC_d : \left| \dfrac{\cost(\bPi(X), \tau(c))}{\cost(X, c)} - 1\right| \leq \eps  \right] \geq 0.9 \]
In Appendix~\ref{sec:forall}, we expand on why such a statement cannot give non-trivial bounds for a natural association (which associates subspaces spanned by the same dataset points before and after the projection). Namely, we show that, for two natural problems, the $1$-medoid problem and the $1$-column subset selection problem, where such a ``for all'' statement should apply, a Theorem~\ref{thm:main} with $t = o(\log n / \eps^2)$ is not true.

{\color{red} Erik: maybe makes sense to put the precise theorem statements, since it will help reviewers understand the subtlety in ``for all'' statements?}}

\ignore{\textbf{$(k, z)$-Clustering with $\ell_z$-Cost}. These problems are the $(k,0)$-projective clustering problems with $\ell_z$-cost, which capture the $k$-median and $k$-means problems. The optimization problem seeks to minimize, over an arbitrary set $C$ of $k$ center points $c_1,\dots, c_k \in \R^d$, the $\ell_z$-norm of the $n$-dimensional vector where the coordinate $i \in [n]$ encodes the distance between $x_i$ and the closest point in $C$. For this problem, we show in Theorem~\ref{thm:center-based} that dimension $t = O((\log k + z\log(1/\eps))/\eps^2)$ suffices. This improves on an analysis of \cite{MMR19} who give a bound of $O( (\log k + z\log(1/\eps) + z^2) / \eps^2)$ (however, we note that the result of \cite{MMR19} gives a stronger ``for all'' guarantee -- see Remark~\ref{rem:for-all}).

Finally, one would hope for a general theorem relating coresets and the Johnson-Lindenstrauss transform. However, our technique is not entirely a black-box application of coreset algorithms. We will need to use specific implementations of the coreset algorithms, as well as additional geometric facts that are specific to each problem. In Appendix~\ref{sec:forall}, we give two examples of (similar) projective clustering problems and certain instances which do admit small coresets, but one cannot apply the Johnson-Lindenstrauss transform to $o(\log n)$ dimensions without paying significantly in the distortion. The examples are the ``medoid'' and ``column subset selection'' problems, two well-studied center-based clustering and subspace approximation problems (respectively) where one imposes an additional restriction on the candidate solutions $\calC_d$. }

\ignore{
For example, the ``$k$-median problem'' corresponds to $(k, 0)$-projective clustering with $\ell_1$: we divide the dataset into $k$ clusters, and each cluster is approximated by a center (i.e., $0$-dimensional affine subspace), the cost is the sum of distances to the centers (i.e., $\ell_1$ norm of cost vector). Quantitatively we show the following results for the

Another example is ``principal components analysis,'' where for $j$ components is the $(1, j)$-projective clustering with $\ell_2$-norm cost. We give a general technique for analyzing the effect of Johnson-Lindenstrauss transforms on these problems by using algorithms for constructing \emph{coresets} (see the recent survey~\cite{F20}). Quantitatively, we summarize our results in the following table.

\begin{figure}[H]
\begin{center}
\begin{tabular}{ |c||c|c|c|  } 
 \hline
 \multicolumn{3}{|c|}{Target Dimension for Johnson-Lindenstrauss Transforms} \\
 \hline
Problem & New Result & Prior Best \\
 \hline
$(k,z)$-clustering  & $O((\log k + z \log(1/\eps))/\eps^2 )$    & $O((\log k + z\log(1/\eps) + z^2) / \eps^2)$ \small{\cite{MMR19}} \\
$(k,z)$-subspace & $\tilde{O}(zk^2 / \eps^3)$  & $O(k/\eps^2)$ (when $z = 2$) \small{\cite{CEMMP15}} \\
 $(k,z)$-flat & $\tilde{O}(zk^2/\eps^3)$ & $\tilde{O}(zk^2 \log n/\eps^3)$ \small{\cite{KR15}} \\
 $(k,z)$-line    &$O((k\log\log n + z + \log(1/\eps))/\eps^3)$ & None \\
  \hline
\end{tabular}
\end{center}
\caption{The $(k,z)$-clustering problem equates to $(k, 0)$-projective clustering with $\ell_z$ cost, generalizing $k$-median when $z = 1$, $k$-mean when $z = 2$, and $k$-center when $z = \infty$. The $(k,z)$-flat approximation problem is $(1, k)$-projective clustering with $\ell_z$ cost, and $(k,z)$-subspace approximation is $(1,k)$-projective clustering with $\ell_z$ cost when subspaces considered are not affine (i.e., contain $0$). Finally, $(k,z)$-line approximation is the $(k, 1)$-projective clustering with $\ell_z$ cost, generalizing $k$-line median when $z = 1$ and the minimum enclosing cylinders when $z = \infty$.
}
\end{figure}}

\subsection{Related Work} 
\paragraph{Dimension Reduction.} Our paper continues a line of work initiated by Boutsidis, Zouzias, and Drineas \cite{BZD10}, who first studied the effect of a Johnson-Lindenstrauss transform for $k$-means clustering, and showed that $t = O(k/\eps^2)$ sufficed for a $2 + \eps$-approximation. 
The bound was improved to $(1+\eps)$-approximation with $t = O(k/\eps^2)$ by Cohen, Elder, Musco, Musco, Persu~\cite{CEMMP15}, who also showed that $t = O(\log k / \eps^2)$ gave a $9+\eps$-approximation. 
Becchetti, Bury, Cohen-Addad, Grandoni, Schwiegelshohn~\cite{BBCGS19} showed that $t = O((\log k + \log\log n) \log(1/\eps)/ \eps^6)$ sufficed for preserving the costs of all $k$-mean clusterings. Makarychev, Makarychev, and Razenshteyn \cite{MMR19} improved and generalized the above bounds for all $(k,z)$-clustering. They showed that $t = O((\log k + z\log(1/\eps) + z^2) / \eps^2)$ preserved costs to $(1\pm \eps)$-factor.

For subspace approximation problems, \cite{CEMMP15} showed that $t = O(k/\eps^2)$ preserves the cost of $(k,2)$-subspace approximation to $(1+\eps)$-factor. In addition, \cite{KR15} showed that $O(zk^{2} \log n / \eps^3)$ preserved the cost of the $(k,z)$-flat approximation to $(1+\eps)$-factor.


\paragraph{Coresets.} Coresets are a well-studied technique for reducing the size of a dataset, while approximately preserving a particular desired property. Since its formalization in Agarwal, Har-Peled, and Varadarajan \cite{AHV05}, coresets have played a foundational role in computational geometry, and found widespread application in clustering, numerical linear algebra, and machine learning (see the recent survey~\cite{F20}). Indeed, even for clustering problems in Euclidean spaces, there is a long line of work (which is still on-going)~\cite{BHI02, HM04, AHV05, C09, LS10, FL11, VX12, VX12b, FSS13, BFL16, SW18, HV20, BBHJKW20, CSS21, CLSS22} exploring the best coreset constructions.

Most relevant to our work is the ``sensitivity sampling'' framework of Feldman and Langberg~\cite{FL11}, which gives algorithms for constructing coresets for the projective clustering problems we study. In light of the results of \cite{FL11}, as well as the classical formulation of the Johnson-Lindenstrauss lemma \cite{JL84}, it may seem natural to apply coreset algorithms and dimensionality reduction concurrently. However, this is not without a few technical challenges. As we will see in the next subsection, it is not necessarily the case that coreset algorithms and random projections ``commute.'' Put succinctly, the random projection $\bPi$ of a coreset of $X$ may not be a coreset of the random projection $\bPi(X)$. Indeed, proving such a statement constitutes the bulk of the technical work.

 
 \subsection{Organization}
 
 The following section (Section~\ref{sec:overview}) overviews the high-level plan, since all our results follow the same technique. To highlight the technique, the first technical section considers the case of $(k, z)$-clustering (Section~\ref{sec:clustering}), where the technique of arguing via coresets shows to obtain $t = O((\log(k) + z\log(1/\eps))/\eps^2)$. The remaining sections cover the technical material for $(k,z)$-subspace approximation (in Section~\ref{sec:subspace}), $(k,z)$-flat approximation (in Section~\ref{sec:flat}), and finally $(k,z)$-line approximation (in Section~\ref{sec:line}).
 

\renewcommand{\cost}{\mathrm{cost}}

\section{Overview of Techniques}\label{sec:overview}

In this subsection, we give a high-level overview of the techniques employed. As it will turn out, all results in this paper follow from one general technique, which we instantiate for the various problem instances. 

We give an abstract instantiation of the approach. We will be concerned with geometric optimization problems of the following sort:
\begin{itemize}
\item For each $d \in \N$, we specify a class of \emph{candidate solutions} given by a set $\calC_d$. For example, in center-based clustering, $\calC_d$ may be given by a tuple of $k$ points in $\R^d$, corresponding to $k$ centers for a center-based clustering. In subspace approximation, the set $\calC_d$ may denote the set of all $k$-dimensional subspaces of $\R^d$.
\item There will be a cost function $f_{d} \colon \R^d \times \calC_d \to \R_{\geq 0}$ which, takes a point $x \in \R^d$ and a potential solution $c \in \calC_d$, and outputs the cost of $x$ on $c$. Continuing on the example on center-based clustering, $f_{d}$ may denote the distance from a dataset point $x \in \R^d$ to its nearest point in $c$. In subspace approximation, $f_{d}$ may denote the distance from a dataset point to the orthogonal projection of that point onto the subspace $c$. For a parameter $z \in \N$, we will denote the cost of using $c$ for a dataset $X \subset \R^d$ by
\[ \cost_{d,z}(X, c) = \left( \sum_{x \in X} f_{d}(x, c)^z \right)^{1/z}. \]
For simplicity in the notation, we will drop the subscripts from the functions $f$ and $\cost$ when they are clear from context.
\end{itemize}
We let $\calJ_{d,t}$ denote a distribution over linear maps $\bPi \colon \R^d \to \R^t$ which will satisfy some ``Johnson-Lindenstrauss'' guarantees (we will specify in the preliminaries the properties we will need). For concreteness, we will think of $\bPi \sim \calJ_{d,t}$ given by matrix multiplication by a $t \times d$ matrix of i.i.d draws of $\calN(0, 1/t)$. We ask, for a particular bound on the dataset size $n \in \N$, a  geometric optimization problem (specified by $\{ \calC_d \}_{d \in \N}$, $f_d$ and $z$), and a parameter $\eps \in (0,1)$, what is the smallest $t \in \N$ such that with probability at least $0.9$ over the draw of $\bPi \sim \calJ_{d,t}$,
\begin{align} 
\frac{1}{1+\eps} \cdot \min_{c \in \calC_{d}} \cost(X, c) \leq \min_{c \in \calC_{t}} \cost(\bPi(X), c) \leq (1+\eps) \cdot \min_{c \in \calC_d} \cost(X, c). \label{eq:goal}  
\end{align}
The right-most inequality in (\ref{eq:goal}) claims that the cost after applying $\bPi$ does not increase significantly, i.e., $\min_{c \in \calC_t} \cost(\bPi(X), c) \leq (1 +\eps) \min_{c \in \calC_d} \cost(X, c)$. This direction is easy to prove for the following reason. For a dataset $X \subset \R^d$, we consider the solution $c^* \in \calC_d$ minimizing $\cost(X, c^*)$. We sample $\bPi \sim \calJ_{d,t}$ and we find a candidate solution $c^{**} \in \calC_{t}$ which exhibits an upper bound on $\min_{c \in \calC_t} \cost(\bPi(X), c) \leq \cost(\bPi(X), c^{**})$. For example, in the center-based clustering, $c^* \in \calC_d$ is a set of $k$ centers in $\R^d$, and we may consider $c^{**} \in \calC_t$ as the $k$ centers from $c^*$ after applying $\bPi$. The fact that $\cost(\bPi(x), c^{**}) \leq (1+\eps) \cost(X,c^*)$ with high probability over $\bPi \sim \calJ_{d,t}$ will follow straight-forwardly from properties of $\calJ_{d,t}$. Importantly, the optimal solution $c^*$ does not depend on $\bPi \sim \calJ_{d,t}$. In fact, while we expect $\bPi \colon \R^d \to \R^t$ to distort some distances substantially, we can pick $c^{**} \in \calC_t$ so that too many distortions on these points is unlikely. 

However, the same reasoning \emph{does not} apply to the left-most inequality in (\ref{eq:goal}). This is because the solution $c^{**} \in \calC_{t}$ which minimizes $\min_{c \in \calC_t} \cost(\bPi(X), c)$ depends on $\bPi$. Indeed, we would expect $c^{**} \in \calC_t$ to take advantage of distortions in $\bPi$ in order to decrease the cost of the optimal solution. We proceed by the following high level plan. We identify a sequence of important events defined over the draw of $\bPi \sim \calJ_{d,t}$ which occur with probability at least $0.9$. The special property is that if $\bPi$ satisfies these events, we can identify, from $c^{**} \in \calC_t$ minimizing $\cost(\bPi(X), c^{**})$, a candidate solution $c^* \in \calC_d$ which exhibits an upper bound $\cost(X, c^*) \leq (1+\eps) \cost(\bPi(X), c^{**})$. 

We now specify how exactly we define, for an optimal $c^{**} \in \calC_t$ (depending on $\bPi$), a candidate solution $c^* \in \calC_t$ whose cost is not much higher than $\cost(\bPi(X), c^{**})$. For that, we will use the notion of coresets. Before the formal definition, we note there is a natural extension of $\cost$ for weighted datasets. In particular, if $S \subset \R^d$ is a set of points and $w \colon S \to \R_{\geq 0}$ is a set of weights for $S$, then we will use $\cost( (S, w), c)$ as $1/z$-th power of the sum over all $x \in S$ of $w(x) \cdot f_{d}(x, c)^z$. 

\begin{definition}[(Weak)\footnote{The word ``weak'' is used to distinguish these from ``strong'' coresets. These are a weighted subset of points which approximately preserve the cost of all candidate solutions.} Coresets, see also \cite{F20}]\label{def:weak-coreset}
For $d \in \N$, let $\calC_d$ denote a class of candidate solutions and $f \colon \R^d \times \calC_{d} \to \R_{\geq 0}$ specify the cost of a point to a solution. For a dataset $X \subset \R^d$ and a parameter $\eps \in (0, 1)$, a (weak) $\eps$-coreset for $X$ is a weighted set of points $S \subset \R^d$ and $w \colon S \to \R_{\geq 0}$ which satisfy
\[ \dfrac{1}{1+\eps} \cdot \min_{c \in \calC_d} \cost(X, c) \leq \min_{c \in \calC_d} \cost((S, w), c) \leq (1+\eps) \cdot \min_{c \in \calC_d} \cost(X, c). \]
\end{definition}

\newcommand{\ALG}{\texttt{ALG}}

It will be crucial for us that these problems admit small coresets. More specifically, for the problems considered in this paper, there exists (known) algorithms which can produce small-size coresets from a dataset. In what follows, $\ALG$ is a randomized algorithm which receives as input a dataset $X \subset \R^d$ and outputs a weighted subset of points $(\bS, \bw)$ which is a weak $\eps$-coreset for $X$ with high probability. Computationally, the benefit of using coresets is that the sets $\bS$ tend to be much smaller than $X$, so that one may compute on $(\bS, \bw)$ and obtain an approximately optimal solution for $X$. For us, the benefit will come in defining the important events. At a high level, since $\bS$ is small, the important events defined with respect to $\bPi$ will only worry about distortions within the subset (or subspace spanned by) $\bS$.

In particular, it is natural to consider the following approach:
\begin{enumerate}
\item\label{en:step-1} We begin with the original dataset $X \subset \R^d$, and we consider the solution $c \in \calC_d$ which minimizes $\cost(X, c)$. The goal is to show that $\cost(X, c)$ cannot be much larger than $\cost(\bPi(X), c^{**})$, where $c^{**} \in \calC_t$ minimizes $\cost(\bPi(X), c^{**})$.
\item Instead of considering the entire dataset $X$, we execute $\ALG(X)$ and consider the weak $\eps$-coreset $(\bS, \bw)$ that we obtain. If we can identify a candidate solution $c^* \in \calC_d$ whose cost $\cost((\bS, \bw), c^{*}) \leq (1+\eps) \cost(\bPi(X), c^{**})$, we would be done. Indeed, $\min_{c \in \calC_d} \cost((\bS, \bw), c) \leq \cost((\bS, \bw), c^{*})$, and the fact $(\bS, \bw)$ is a weak $\eps$-coreset implies $\min_{c \in \calC_d} \cost(X, c) \leq (1+\eps) \min_{c\in \calC_d}\cost((\bS, \bw), c)$.
\item Moving to a coreset $(\bS, \bw)$ allows one to relate $\cost((\bS, \bw), c^*)$ and $\cost((\bPi(\bS), \bw), c^{**})$ by considering the performance of $\bPi$ on $\bS$. The benefit is that the important events, defined over the draw of $\bPi \sim \calJ_{d,t}$, set $t$ as a function of $|\bS|$, instead of $|X|$. A useful example to consider is requiring $\bPi^{-1}$ be $(1+\eps)$-Lipschitz on the entire subspace spanned by $\bS$, which requires $t = \Theta(|\bS| / \eps^2)$. For the problems considered here, a nearly optimal $c^{**} \in \calC_t$ for $(\bS, \bw)$ will be in the subspace spanned by $\bS$, so we may identify $c^{*} \in \calC_d$ whose cost on $(\bS, \bw)$ is not much higher than the $\cost((\bPi(\bS), \bw), c^{**})$ by evaluating $\bPi^{-1}(c^{**})$ since $c^{**}$ lies inside $\mathrm{span}(\bS)$.\footnote{While the above results in bounds for $t$ which are already meaningful, we will exploit other geometric aspects of the problems considered to get bounds on $t$ which are logarithmic in the coreset size. For center-based clustering, \cite{MMR19} showed that one may apply Kirzbraun's theorem. For subspace approximation, we use the geometric lemmas of \cite{SV12}.}
\item\label{en:step-4} The remaining step is showing $\cost((\bPi(\bS), \bw), c^{**}) \leq (1+\eps) \cost(\bPi(X), c^{**})$. In particular, one would like to claim $(\bPi(\bS), \bw)$ is a weak $\eps$-coreset for $\bPi(X)$ and use the right-most inequality in Definition~\ref{def:weak-coreset}. However, it is not clear this is so. The problem is that the algorithm $\ALG$ depends on the $d$-dimensional representation of $X \subset \R^d$, and $(\bPi(\bS), \bw)$ may not be a valid output for $\ALG(\bPi(X))$. As we show, this does work for (some) coreset algorithms built on the \emph{sensitivity sampling} framework (see,~\cite{FL11, BFL16}).\footnote{We will not prove that with high probability over $\bPi$ and $\ALG(X)$, $(\bPi(\bS), \bw)$ is a weak $\eps$-coreset for $\bPi(X)$. Rather, all we need is that the right-most inequality in Definition~\ref{def:weak-coreset} holds for $(\bPi(\bS), \bw)$ and $\bPi(X)$, which is what we show.}
\end{enumerate}

\subsubsection{Sensitivity Sampling for Step~\ref{en:step-4}} In the remainder of this section, we briefly overview the sensitivity sampling framework, and the components required to make Step~\ref{en:step-4} go through. At a high level, coreset algorithms in the sensitivity sampling framework proceed in the following way. Given a dataset $X \subset \R^d$, the algorithm computes a \emph{sensitivity sampling distribution} $\tilde{\sigma}$ supported on $X$. The requirement is that, for each potential solution $c \in \calC_d$, sampling from $\tilde{\sigma}$ gives a low-variance estimator for $\cost_{d,z}(X, c)^z$. In particular, we let $\tilde{\sigma}(x)$ be the probability of sampling $x \in X$. Then, for any distribution $\tilde{\sigma}$ and any $c \in \calC_{d}$,
\begin{align}
\Ex_{\bx \sim \tilde{\sigma}}\left[ \dfrac{1}{\tilde{\sigma}(\bx)} \cdot \dfrac{f_d(\bx, c)^z}{\cost_{d,z}(X,c)^z} \right] = 1. \label{eq:expectation} 
\end{align}
Equation~\ref{eq:expectation} implies that, for any $m \in \N$, if $\bS$ is $m$ i.i.d samples from $\tilde{\sigma}$ and $\bw(x) = 1 / (m \tilde{\sigma}(x))$, the expectation of $\cost_{d,z}((\bS, \bw), c)^z$ is $\cost_{d,z}(X, c)^z$. In addition, the algorithm designs $\tilde{\sigma}$ so that, for a parameter $T > 0$, 
\begin{align} 
\sup_{c \in \calC_d} \Ex_{\bx \sim \tilde{\sigma}}\left[ \left( \dfrac{1}{\tilde{\sigma}(\bx)} \cdot \dfrac{f_d(\bx, c)^z}{\cost_{d,z}(X, c)^z}\right)^2 \right]  \leq T. \label{eq:variance}
\end{align}
If we set $m \geq T / \eps^2$, (\ref{eq:variance}) and Chebyshev's inequality implies $\cost_{d,z}((\bS, \bw), c)^z \approx_{1\pm\eps} \cost_{d,z}(X, c)^z$ for each $c \in \calC_{d}$ with a high constant probability, and the remaining work is in increasing $m$ by a large enough factor to ``union bound'' over all $c \in \calC_{d}$. There is a canonical way of ensuring $\tilde{\sigma}$ and $T$ satisfy (\ref{eq:variance}): we first define $\sigma \colon X \to \R_{\geq 0}$, known as a ``sensitivity function'', which sets for each $x \in X$,
\begin{align}
\sigma(x) \geq \sup_{c \in \calC_d} \dfrac{f_d(x, c)^z}{\cost_{d,z}(X, c)^{z}}, \qquad\text{and}\qquad T = \sum_{x \in X} \sigma(x),\label{eq:sensitivity-function}
\end{align}
which is known as the ``total sensitivity.'' Then, the distribution is given by letting $\tilde{\sigma}(x) = \sigma(x) / T$. 

We now show how to incorporate the map $\bPi \sim \calJ_{d,t}$, to argue Step~\ref{en:step-4}. Recall that we let $\bS$ denote $m$ i.i.d draws from $\tilde{\sigma}$ and the weights be $\bw(x) = 1 / (m \tilde{\sigma}(x))$. We want to argue that, with high constant probability over the draw of $(\bS, \bw)$ and $\bPi \sim \calJ_{d,t}$, we have
\begin{align} 
\cost_{t,z}((\bPi(\bS), \bw), c^{**}) \leq (1+\eps) \cdot \cost_{t,z}(\bPi(X), c^{**}). \label{eq:step-4-goal}
\end{align}
First, note that the analogous version of (\ref{eq:expectation}) for $\cost_{t,z}(\bPi(X), c)$ continues to hold. In particular, for any map $\bPi$ in the support of $\calJ_{d,t}$ and $c^{**} \in \calC_t$ minimizing $\cost_{t,z}(\bPi(X), c)$, 
\begin{align} 
\Ex_{\bx \sim \tilde{\sigma}}\left[\dfrac{1}{\tilde{\sigma}(\bx)} \cdot \dfrac{f_{t}(\bPi(\bx), c^{**})^z}{\cost_{t,z}(\bPi(X), c^{**})^z} \right] = 1. \label{eq:expectation-after-Pi}
\end{align}
Hence, it remains to define $\tilde{\sigma}$ satisfying (\ref{eq:sensitivity-function}) which also satisfies one additional requirement. With high probability over $\bPi \sim \calJ_{d,t}$, we should have
\begin{align} 
\Ex_{\bx \sim \tilde{\sigma}}\left[ \left(\dfrac{1}{\tilde{\sigma}(\bx)} \cdot \dfrac{f_{t}(\bPi(\bx), c^{**})^z}{\cost_{t,z}(\bPi(X), c^{**})}  \right)^2\right] \lsim T. \label{eq:variance-after-map}
\end{align}
The above translates to saying, for most $\bPi \sim \calJ_{d,t}$ the variance of $\cost((\bPi(\bS), \bw), c^{**})$, when $m = O(T/\eps^2)$, is small. Once that is established, we may apply Chebyshev's inequality and conclude (\ref{eq:step-4-goal}) with high constant probability.\footnote{Since Steps~\ref{en:step-1}--\ref{en:step-4} only argued about the optimal $c^{**} \in \calC_{t}$, there is no need to ``union bound'' over all $c \in \calC_t$ in our arguments.} 

\subsubsection{The Circularity and How to Break It} One final technical hurdle arises. While one may define the sensitivity function $\sigma(x)$ to be exactly $\sup_{c \in \calC_d} f_d(x, c)^z / \cost_{d,z}(X,c)^z$ and automatically satisfy (\ref{eq:sensitivity-function}), it becomes challenging to argue that (\ref{eq:variance-after-map}) holds. In the end, the complexity we seek to optimize is the total sensitivity $T$, so there is flexibility in defining $\sigma$ while showing (\ref{eq:variance-after-map}) holds. In fact, sensitivity functions $\sigma$ which are computationally simple tend to be known, since an algorithm using coresets must quickly compute $\sigma(x)$ for every $x \in X$. 

The sensitivity functions $\sigma$ used in the literature (for instance, in \cite{FL11, VX12}) are defined with respect to an approximately optimal $c \in \calC_{d}$ (or bi-criteria approximation) for $\cost_{d,z}(X, c)$. Furthermore, the arguments used to show these function satisfy (\ref{eq:sensitivity-function}), which we will also employ for  (\ref{eq:variance-after-map}), crucially utilize the approximation guarantee on $c \in \calC_{d}$. The apparent circularity appears in approximation algorithms and also shows up in the analysis here:
\begin{itemize}
\item For $X \subset \R^d$, we identify the optimal $c \in \calC_{d}$ minimizing $\cost_{d,z}(X, c)$, and use $c$ to define $\sigma \colon X \to \R_{\geq 0}$. The fact that $c \in \calC_d$ is optimal (and therefore approximately optimal) allows us to use known arguments (in particular, those in \cite{VX12, VX12b}) to establish (\ref{eq:sensitivity-function}) and give an upper bound on $T$.
\item We use the proof of the ``easy'' direction to identify a solution $c' \in \calC_t$ with $\cost_{t,z}(\bPi(X), c') \leq (1+\eps) \cost_{d,z}(X, c)$ (recall this was used to establish the right-most inequality in (\ref{eq:goal})). From the analytical perspective, it is useful to think of $\sigma' \colon \bPi(X) \to \R_{\geq 0}$ as the function one would get from defining a sensitivity function like in the previous step with $c'$ instead of $c$. If we could show $c' \in \calC_t$ was approximately optimal for $\bPi(X)$, we could use \cite{VX12, VX12b} again to argue (\ref{eq:variance-after-map}). The circularity is the following. Showing $c' \in \calC_t$ is approximately optimal means showing an upper bound on $\min_{c \in \calC_t} \cost_{t,z}(\bPi(X), c)$ in terms of $\cost_{t,z}(\bPi(X),c')$. Since, we picked $\cost_{t,z}(\bPi(X), c')$ to be at most $ (1+\eps)\min_{c \in \calC_{d}} \cost_{d,z}(X,c)$, this is exactly what we sought to prove.
\end{itemize}
If ``approximately optimal'' above required $c' \in \calC_{t}$ be a $(1+\eps)$-approximation to the optimal $\bPi(X)$, we would have a complete circularity and be unable to proceed. However, similarly to the case of approximation algorithms, it suffices to have a poor approximation. Suppose we showed $c' \in \calC_{t}$ was a $C$-approximation, then, increasing $m$ by a factor depending on $C$ (which would affect the resulting dimensionality $t$) would account for this increase and drive the variance back down to $\eps^2$. Moreover, showing $c'$ is a $O(1)$-approximation with probability at least $0.99$ over $\bPi \sim \calJ_{d,t}$, given Steps~\ref{en:step-1}--\ref{en:step-4} is straight-forward. Instead of showing the stronger bound that $\cost_{t,z}((\bPi(\bS), \bw), c^{**}) \leq (1+\eps) \cost_{t,z}(\bPi(X), c^{**})$, we show that $\cost_{t,z}((\bPi(\bS), \bw), c^{**}) \leq O(1) \cdot \cost_{t,z}(\bPi(X), c^{**})$. The latter (loose) bound is a consequence of applying Markov's inequality to (\ref{eq:expectation-after-Pi}).

In summary, we overcome the circularity by going through Steps~\ref{en:step-1}--\ref{en:step-4} twice. In the first time, we show a weaker $O(1)$-approximation. Specifically, we show that $\bPi\sim \calJ_{d,t}$ preserves the cost of $\min_{c \in \calC_t} \cost_{d,z}(\bPi(X), c)$ up to factor $O(1)$. The first time around, we won't upper bound the variance in (\ref{eq:variance-after-map}), and we simply use Markov's inequality to (\ref{eq:expectation-after-Pi}) in order to prove a (loose) bound on Step~\ref{en:step-4}. Once we've established the $O(1)$-factor approximation, we are guaranteed that $c' \in \calC_t$ is a $O(1)$-approximation to $\min_{c \in \calC_t} \cost_{t,z}(\bPi(X), c)$. This means that, actually, the sensitivity sampling distribution $\tilde{\sigma}$ we had considered (when viewed as a sensitivity sampling distribution for $\bPi(X)$) gives estimators with bounded variance, as in (\ref{eq:variance-after-map}). In particular, going through Steps~\ref{en:step-1}--\ref{en:step-4} once again implies that $c' \in \calC_t$ was actually the desired $1\pm \eps$-approximation.

\section{Preliminaries}

We specify the properties we use from the distribution $\calJ_{d,t}$. We will refer to these as ``Johnson-Lindenstrauss'' distributions. Throughout the proof, we will often refer to $\calJ_{d,t}$ as given by a $t\times d$ matrix of i.i.d draws from $\calN(0, 1/t)$. The goal of specifying the useful properties is to use other ``Johnson-Lindenstrauss''-like distributions. The first property we need is that $\bPi \colon \R^d \to \R^t$ is a linear map, and that any $x, y \in \R^d$ satisfies
\[ \Ex_{\bPi \sim \calJ_{d,t}}\left[ \dfrac{\| \bPi(x) - \bPi(y)\|_2^2}{\|x - y\|_2^2}\right] = 1.\]
We use the standard property of $\calJ_{d,t}$, that $\bPi$ preserves distances with high probability, i.e., for any $x, y \in \R^{d}$,
\[ \Prx_{\bPi \sim \calJ_{d,t}}\left[ \left| \dfrac{\|\bPi(x) - \bPi(y)\|_2^2}{\|x-y\|_2^2} - 1 \right| \geq \eps \right] \leq e^{-\Omega(\eps^2 t)}. \]
More generally, we use the conclusion of the following lemma. We give a proof when $\calJ_{d,t}$ is a $t\times d$ i.i.d entries of $\calN(0, 1/t)$.
\begin{lemma}\label{lem:gaussian-guarantee}
Let $\calJ_{d,t}$ denote a Johnson-Lindenstrauss distribution over maps $\R^d \to \R^t$ given by a matrix multiplication on the left by a $t\times d$ matrix of i.i.d draws of $\calN(0, 1/t)$. If $t \gsim z / \eps^2$, then for any $x, y \in \R^d$,
\begin{align*}
\Ex_{\bPi \sim \calJ_{d.t}}\left[ \left( \dfrac{\| \bPi(x) - \bPi(y)\|_2^z}{\|x - y\|_2^z} - 1 \right)^+ \right] \leq \dfrac{(1+\eps)^z - 1}{100}.
\end{align*}
\end{lemma}

\begin{proof}
We note that by the 2-stability of the Gaussian distribution, we have $\| \bPi(x) - \bPi(y)\|_2^2$ is equivalently distributed as $\|\bg\|_2^2 \cdot \|x - y\|_2^2$, where $\bg \sim \calN(0, I_t / t)$. Therefore, we have that for any $\lambda > 0$ which we will optimize shortly to be a small constant times $\eps$,
\begin{align*}
\Ex_{\bPi \sim \calJ(d,t)}\left[ \left(\dfrac{\|\bPi(x) - \bPi(y)\|_2^z}{\|x-y\|_2^z} - 1 \right)^+\right] &\leq (1+\lambda)^z - 1 + \Ex_{\bg \sim \calN(0, I_t/t)}\left[ (\| \bg\|_2^z - (1+\lambda)^z )^+\right].
\end{align*}
Furthermore, 
\begin{align*}
\Ex_{\bg \sim \calN(0, I_t/t)}\left[ \left( \| \bg \|_2^z - (1+\lambda)^z \right)^+\right] &= \int_{u:(1+\lambda)^z}^{\infty} \Prx\left[ \| \bg\|_2^z \geq u \right] du = \frac{z}{2} \int_{v:(1+\lambda)^2}^{\infty} \Prx\left[ \|\bg\|_2^2 \geq v\right] \cdot v^{z/2 - 1} dv.
\end{align*}
We will upper bound the probability that $\|\bg\|_2^2$ exceeds $v$ by the Chernoff-Hoeffding method. In particular, recall that $\|\bg\|_2^2$ when $\calN(0, I_t/t)$ is distributed as a $\chi^2$-random variable with $t$ degrees of freedom, rescaled by $1/t$, such that the moment generating function of $\|\bg\|_2^2$ has the following closed form solution whenever $\alpha < t/2$:
\[ \log\left(\Ex_{\bg \sim \calN(0, I_t/t)}\left[ \exp\left( \alpha \| \bg\|_2^2\right)\right]\right) = -\frac{t}{2} \log\left(1 - \frac{2\alpha}{t} \right) \leq \alpha + \frac{2\alpha^2}{t}  \]
In particular, for any $\alpha < t / 2$, we may upper bound
\begin{align*}
\frac{z}{2} \int_{v:(1+\lambda)^2}^{\infty} \Prx\left[ \|\bg\|_2^2 \geq v\right] v^{z/2 -1} &\leq \frac{z}{2} \cdot \exp\left(\alpha + \frac{2\alpha^2}{t} \right)\int_{v:(1+\lambda)^2}^{\infty} \exp\left( - \left(\alpha - \frac{z}{2} - 1 \right) v \right) dv \\
	&\leq \dfrac{z}{2\alpha - z - 2} \cdot \exp\left( \frac{2\alpha^2}{t} - \alpha \lambda + (1+\lambda) (z/2+ 1) \right) < \frac{\lambda}{100},
\end{align*}
by letting setting $\alpha = t \lambda / 10$ whenever $t \gsim z / \lambda^2$. Setting $\lambda$ to be a small constant of $\eps$ gives the desired guarantees.
\end{proof}

\begin{definition}[Subspace Embeddings]
Let $d \in \N$ and $A \subset \R^d$ denote a subspace of $\R^d$. For $\eps > 0$, a map $f \colon \R^d \to \R^t$ is an $\eps$-subspace embedding of $A$ if, for any $x \in A$,
\[ \frac{1}{1+\eps} \cdot \|f(x)\|_2 \leq \| x\|_2 \leq (1+\eps) \cdot \|f(x)\|_2. \]
\end{definition}

\begin{lemma}\label{lem:subspace-embedding}
Let $d \in \N$ and $A \subset \R^d$ be a subspace of dimension at most $k$. For $\eps, \delta \in (0,1/2)$, let $\calJ_{d,t}$ denote a Johnson-Lindenstrauss distribution over maps $\R^d \to \R^t$ given by a matrix multiplication on the left by a $t\times d$ matrix of i.i.d draws of $\calN(0, 1/t)$. If $t \sim (k + \log(1/\delta)) / \eps^2$, then $\bPi \sim \calJ_{d,t}$ is an $\eps$-subspace embedding of $A$ with probability at least $1 - \delta$. 
\end{lemma}


\section{Center-based $(k,z)$-clustering}\label{sec:clustering}

\renewcommand{\cost}{\mathrm{cost}}

In the $(k,z)$-clustering problems, for any set $C \subset \R^d$ of $k$ points, and point $x \in \R^d$, we write
\[ \cost_z^z(x, C) = \min_{c \in C} \| x - c\|_2^2,\]
and for a subset $X \subset \R^d$,
\[ \cost_z^z(X, C) = \sum_{x \in X} \cost(x, C) = \sum_{x \in X} \min_{c \in C} \| x - c\|_2^z. \]
We extend the above notation to weighted subsets, where for a subset $S \subset \R^d$ with (non-negative) weights $w \colon S \to \R_{\geq 0}$, we write $\cost_z^z((S, w), C) = \sum_{x \in S} w(x) \min_{c \in C} \|x - c \|_2^z$. The main result of this section is the following theorem.

\begin{theorem}[Johnson-Lindenstrauss for Center-Based Clustering]\label{thm:center-based}
Let $X = \{ x_1,\dots, x_n \} \subset \R^d$ be any set of points, and let $C \subset \R^d$ denote the optimal $(k,z)$-clustering of $X$. For any $\eps \in (0, 1/2)$, suppose we let $\calJ_{d,t}$ be the distribution over Johnson-Lindenstrauss maps where
\[ t \gsim \frac{\log k + z\log(1/\eps)}{\eps^2},\]
 Then, with probability at least $0.9$ over the draw of $\bPi \sim \calJ_{d,t}$, 
\[ \frac{1}{1+\eps} \cdot \cost_z(X, C) \leq  \min_{\substack{C' \subset \R^t \\ |C'| = k}} \cost_z(\bPi(X), C') \leq (1+\eps) \cost_z(X, C). \] 
\end{theorem}

There are two directions to showing dimension reduction: (1) the optimal clustering in the \emph{reduced space} is not too expensive, and (2) the optimal clustering in the reduced space is not too cheap. We note that (1) is simple because we can exhibit a clustering in the reduced space whose cost is not too high; however, (2) is much trickier, since we need to rule out a too-good-to-be-true clustering in the reduced space.

\subsection{Easy Direction: Optimum Cost Does Not Increase}

\begin{lemma}\label{lem:upper-bound}
Let $X = \{ x_1,\dots, x_n \} \subset \R^d$ be any set of points and let $C \subset \R^d$ of size $k$ be the centers of the optimal $(k,z)$-clustering of $X$. We let $\calJ_{d,t}$ be the distribution over Johnson-Lindenstrauss maps. If $t \gsim z/\eps^2$, then with probability at least $0.99$ over the draw of $\bPi \sim \calJ_{d,t}$, 
\[ \left( \sum_{x \in X} \| \bPi(x) - \bPi(c(x))\|_2^z \right)^{1/z} \leq (1+\eps) \cost_z(X, C), \]
and hence,
\[ \min_{\substack{C' \subset \R^t \\ |C'| = k}} \cost_z(\bPi(X), C') \leq (1+\eps) \min_{\substack{C \subset \R^d \\ |C| = k}} \cost_z(X, C).   \]  
\end{lemma}

\begin{proof}
For $x \in X$, let $c(x) \in C$ denote the closest point from $C$ to $x$. We compute the expected positive deviation from assigning $\bPi(x)$ to $\bPi(c(x))$, and note that the $\cost_z(\bPi(X), \bPi(C))$ can only be lower. Hence, if we can show
\begin{align}
\dfrac{1}{\cost_z^z(X, C)} \Ex_{\bPi \sim \calJ_{d,t}}\left[ \sum_{x\in X}\left( \left\| \bPi(x) - \bPi(c(x))\right\|_2^z - \| x - c(x) \|_2^z\right)^+ \right] \lsim \dfrac{(1+\eps)^z - 1}{100}, \label{eq:desired}
\end{align}
then by Markov's inequality, we will obtain $\cost_z^z(\bPi(X), \bPi(C)) \leq (1+\eps)^z \cost_z^z(X, C)$, and obtain the desired bound when raising to power $1/z$. This last part follows from Lemma~\ref{lem:gaussian-guarantee}.
\end{proof}

\subsection{Hard Direction: Optimum Cost Does Not Decrease} 

We now show that after applying a Johnson-Lindenstrauss map, the cost of the optimal clustering in the dimension-reduced space is not too cheap. This section will be significantly more difficult, and will draw on the following preliminaries.

\subsubsection{Preliminaries} 

\begin{definition}[Weak and Strong Coresets]\label{def:clustering-coreset}
Let $X = \{ x_1,\dots, x_n \} \subset \R^d$ be a set of points. A (weak) $\eps$-\emph{coreset} of $X$ for $(k,z)$-clustering is a subset $S \subset \R^d$ of points with weights $w \colon S \to \R_{\geq 0}$ such that, 
\[  \dfrac{1}{1+\eps} \cdot \min_{\substack{C \subset \R^d \\ |C| \leq k}} \cost_z(X, C) \leq \min_{\substack{C \subset \R^d \\ |C| \leq k}} \cost_z((S,w), C) \leq (1+\eps) \cdot  \min_{\substack{C \subset \R^d \\ |C| \leq k}} \cost_z(X, C). \]
The coreset $(S, w)$ is a strong $\eps$-coreset if, for all $C = \{ c_1,\dots, c_k\} \subset \R^d$, we have
\[ \dfrac{1}{1+\eps} \cdot \cost_z(X, C) \leq \cost_z((S, w), C)  \leq (1+\eps) \cdot \cost_z(X, C).  \]
\end{definition}

Notice that Definition~\ref{def:clustering-coreset} gives an approach to finding an approximately optimal $(k,z)$-clustering. Given $X \subset \R^d$, we find a (weak) $\eps$-coreset $(S, w)$ and find the optimal clustering $C \subset \R^d$ with respect to the coreset $(S, w)$. Then, we can deduce that a clustering which is optimal for the coreset is also approximately optimal for the original point set. A common and useful framework for building coresets is by utilizing the ``sensitivity'' sampling framework.

\newcommand{\frakS}{\mathfrak{S}}

\begin{definition}[Sensitivities]
Let $n, d \in \N$, and consider any set of points $X = \{ x_1,\dots, x_n \} \subset \R^d$, as well as $k \in \N$, $z \geq 1$. A sensitivity function $\sigma \colon X \to \R_{\geq 0}$ for $(k,z)$-clustering $X$ in $\R^d$ is a function satisfying, that for all $x \in X$,
\begin{align*}
\sup_{\substack{C \subset \R^d \\ |C| \leq k}} \dfrac{\cost_z^z(x, C)}{\cost_z^z(X, C)}  \leq \sigma(x).
\end{align*}
The total sensitivity of the sensitivity function $\sigma$ is given by
\[ \frakS_{\sigma} = \sum_{x \in X} \sigma(x). \]
For a sensitivity function, we let $\tilde{\sigma}$ denote the sensitivity sampling distribution, supported on $X$, which samples $x \in X$ with probability proportional to $\sigma(x)$. 
\end{definition}

The following lemma gives a particularly simple sensitivity sampling distribution, which will be useful for analyzing our dimension reduction procedure. The proof below will follow from two applications of the triangle inequality which we reproduce from Claim~5.6 in \cite{HV20}.
\begin{lemma}\label{lem:sens}
Let $n, d \in \N$ and consider a set of points $X = \{ x_1,\dots, x_n \} \subset \R^d$. Let $C \subset \R^d$ of size $k$ be optimal $(k,z)$-clustering of $X$, and let $c\colon X \to C$ denote the function which sends $x \in X$ to its closest point in $C$, and let $X_{c} \subset X$ be the set of points where $c(x) = c$. Then, the function $\sigma \colon X \to \R_{\geq 0}$ given by
 \[ \sigma(x) = 2^{z-1} \cdot \dfrac{\|x - c(x) \|_2^z}{\cost_z^z(X, C)} + \dfrac{2^{2z-1}}{|X_{c(x)}|} \]
is a sensitivity function for $(k,z)$-clustering $X$ in $\R^d$, satisfying
\[ \frakS_{\sigma} = 2^{z-1} + 2^{2z-1} \cdot k\]
\end{lemma}

\begin{proof}
Consider any set $C' \subset \R^d$ of $k$ points, and let $c' \colon X \to C'$ be the function which sends each $x \in X$ to its closest point in $C'$. Then, we have
\begin{align}
\| x - c'(x)\|_2^z &\leq \left( \|x - c(x) \|_2 + \| c(x) - c'(x)\|_2\right)^z \leq 2^{z-1} \| x - c(x) \|_2^z + 2^{z-1} \| c(x) - c'(x) \|_2^z \label{eq:sens-1}\\
			&\leq 2^{z-1} \|x - c(x) \|_2^z + \dfrac{2^{z-1}}{|X_{c(x)}|} \sum_{x \in X} \|c(x) - c'(x)\|_2^z  \label{eq:sens-2} \\
			&\leq 2^{z-1} \|x - c(x)\|_2^z + \frac{2^{2(z-1)}}{|X_{c(x)}|} \left( \cost_z^z(X, C) + \cost_z^z(X, C') \right), \label{eq:sens-3}
\end{align}
where we used the triangle inequality and H\"{o}lder inequality in (\ref{eq:sens-1}), and added additional non-negative terms in (\ref{eq:sens-2}), and we finally apply the triangle inequality and H\"{o}lder's inequality once more in (\ref{eq:sens-3}). Hence, using the fact that $C$ is the optimal clustering, we have
\begin{align*}
\dfrac{\cost_z^z(x, C')}{\cost_z^z(X, C')} &\leq 2^{z-1} \cdot \dfrac{\|x - c(x)\|_2^z}{\cost(X, C')} + \dfrac{2^{2(z-1)}}{|X_{c(x)}|} \left( \dfrac{\cost_z^z(X, C)}{\cost_z^z(X, C')} + 1\right) \leq 2^{z-1} \cdot \dfrac{\|x - c(x)\|_2^z}{\cost_z^z(X, C)} + \dfrac{2^{2z - 1}}{|X_{c(x)}|}.
\end{align*}
The bound on $\frakS_{\sigma}$ follows from summing over all $x \in X$, noting the fact that $\sum_{x \in X} 1/|X_{c(x)}| = k$.
\end{proof}

The main idea behind the sensitivity sampling framework for building coresets is to sample from a sensitivity sampling distribution enough times in order to build a coreset. For this work, it will be sufficient to consider the following theorem of~\cite{HV20}, which shows that $\poly(k, 1/\eps^{z})$ draws from an appropriate sensitivity sampling distribution suffices to build strong $\eps$-coresets for $(k,z)$-clustering in $\R^d$.

\begin{theorem}[$\eps$-Strong Coresets from Sensitivity Sampling~\cite{HV20}]\label{thm:strong-coresets}
For any subset $X = \{ x_1,\dots, x_n \} \subset \R^d$ and $\eps \in (0,1/2)$. Let $C \subset \R^d$ of size $k$ be the optimal $(k,z)$-clustering of $X$, and let $\tilde{\sigma}$ denote the sensitivity sampling distribution of Lemma~\ref{lem:sens}.
\begin{itemize}
\item Let $(\bS, \bw)$ denote a random (multi-)set $\bS \subset X$ and $\bw \colon \bS \to \R_{\geq 0}$ given by, for $m = \poly(k, 1/\eps^z)$ iterations, sampling $\bx \sim \tilde{\sigma}$ i.i.d and letting $\bw(\bx) = 1/(m\tilde{\sigma}(x))$.
\item Then, with probability $1 - o(1)$ over the draw of $(\bS, \bw)$, it is an $\eps$-strong coreset for $X$.
\end{itemize}
\end{theorem}

\begin{theorem}[Kirszbraun theorem]
Let $Y \subset \R^{d_1}$ and $\phi \colon Y \to \R^{d_2}$ be an $L$-Lipschitz map (with respect to Euclidean norms on $\R^{d_1}$ and $\R^{d_2}$). There exists a map $\tilde{\phi} \colon \R^{d_1} \to \R^{d_2}$ which is $L$-Lipschitz extending $\phi$, i.e., $\phi(x) = \tilde{\phi}(x)$ for all $x \in Y$.
\end{theorem}

\subsection{The Important Events} 

We now define the important events which will allow us to prove that the optimum $(k,z)$-clustering after dimension reduction does not decrease substantially. We first define the events, and then we prove that if the events are all satisfied, then we obtain our desired lower bound.

\begin{definition}[The Events]\label{def:events}
Let $X = \{ x_1,\dots, x_n \} \subset \R^d$ and $C \subset \R^d$ of size $k$ be centers for an optimal $(k,z)$-clustering of $X$, and $\tilde{\sigma}$ is the sensitivity sampling distribution of $X$ with respect to $C$ as in Lemma~\ref{lem:sens}. We will consider the following experiment, 
\begin{enumerate}
\item We generate a sample $(\bS, \bw)$ by sampling from $\tilde{\sigma}$ for $m = \poly(k, 1/\eps^{z})$ i.i.d iterations $\bx \sim \tilde{\sigma}$ and set $\bw(\bx) = 1/ (m\tilde{\sigma}(\bx))$.
\item Furthermore, we sample $\bPi \sim \calJ_{d, t}$ which is a Johnson-Lindenstrauss map $\R^d \to \R^t$. 
\item We let $\bS' = \bPi(\bS) \subset \R^t$ denote the image of $\bPi$ on $\bS$.
\end{enumerate}
The events are the following:
\begin{itemize}
\item $\bE_1$ : The weighted (multi-)set $(\bS, \bw)$ is a weak $\eps$-coreset for $(k,z)$-clustering $X$ in $\R^d$.
\item $\bE_2$ : The map $\bPi \colon \bS \to \bS'$, given by restricting $\bPi$ is $(1+\eps)$-bi-Lipschitz.
\item $\bE_3(\beta)$ : We let $\bC' \subset \R^t$ of size $k$ be the optimal centers for $(k,z)$-clustering $\bPi(X)$ in $\R^t$. The weighted (multi-)set $(\bPi(\bS), \bw)$ satisfies 
\[ \cost_z^z((\bPi(\bS), \bw), \bC') \leq \beta \cdot \cost_z^z(\bPi(X), \bC'). \]
\end{itemize}
\end{definition}

\begin{lemma}\label{lem:combine-events}
Let $X = \{ x_1,\dots, x_n \} \subset \R^d$, and suppose $(\bS, \bw)$ and $\bPi \colon \R^d \to \R^t$ satisfy events $\bE_1,\bE_2$ and $\bE_3(\beta)$, then, 
\begin{align*}
\min_{\substack{C' \subset \R^t \\ |C'| = k}} \cost_z(\bPi(X), C') \geq \dfrac{1}{\beta^{1/z} (1+\eps)^{1+1/z}} \cdot \min_{\substack{C \subset \R^d \\ |C| = k}} \cost_z(X, C).
\end{align*}
\end{lemma}

\begin{proof}
Let $\bC' \subset \R^t$ of size $k$ denote the centers which give the optimal $(k,z)$-clustering of $\bPi(X)$ in $\R^t$. Then, by $\bE_3$, 
\begin{align*}
\cost_z^z(\bPi(X), \bC') \geq (1/\beta) \cdot \cost_z^z((\bPi(\bS), \bw), \bC').
\end{align*}
Now, we use Kirszbraun's Theorem, and extend $\bPi^{-1} \colon \R^t \to \R^d$ in a $(1+\eps)$-Lipschitz manner. Hence, 
\[ \cost_z((\bPi(\bS), \bw), \bC') \geq \dfrac{1}{1+\eps} \cdot \cost_z((\bS, \bw), \bPi^{-1}(\bC')) \geq \dfrac{1}{1+\eps} \cdot \min_{\substack{C'' \subset \R^d \\ |C''| = k}} \cost_z((\bS, \bw), C''). \]
Finally, using the fact that $(\bS, \bw)$ is an $\eps$-weak coreset, we may conclude that 
\[ \cost_z^z((\bS, \bw), C'') \geq \dfrac{1}{1+\eps} \cdot \cost_z^z(X, C'').   \]
\end{proof}

We now turn to showing that an appropriate setting of parameters implies that the events occur often. For the first event, Theorem~\ref{thm:strong-coresets} from \cite{HV20} implies event $\bE_1$ occurs with probability $1-o(1)$. We state the usual guarantees of the Johnson-Lindenstrauss transform, which is what we need for event $\bE_2$ to hold. 
\begin{lemma}\label{lem:jl}
Let $S \subset \R^d$ be any set of $m$ points, and $\calJ_{d, t}$ denote the Johnson-Lindenstrauss map, with
\[ t \gsim \dfrac{\log m}{\eps^2}. \]
Then, with probability $0.99$ over the draw of $\bPi \sim \calJ_{d,t}$, $\bPi \colon S \to \bPi(S)$ is $(1\pm\eps)$-bi-Lipschitz, and hence, $\bE_2$ occurs with probability $0.99$.
\end{lemma}

\subsubsection{A Bad Approximation Guarantee} 

\begin{lemma}[Warm-Up Lemma]\label{lem:bad-approx}
Fix any $\Pi \in \calJ_{d,t}$ and let $C' \subset \R^t$ denote the optimal centers for $(k,z)$-clustering of $\Pi(X)$, then with probability at least $0.99$ over the draw of $(\bS, \bw)$ as per Definition~\ref{def:events},
\begin{align*}
\sum_{x \in \bS} \bw(x) \cdot \min_{c \in C'} \| \Pi(x) - c\|_2^z \leq 100 \cdot \cost_z^z(\Pi(x), C'),
\end{align*}
with probability at least $0.99$. In other words, $\bE_3(100)$ holds with probability at least $0.99$.
\end{lemma}

\begin{proof}
For any $x \in X$, let $c'(x) \in C'$ denote the point in $C'$ closest to $\Pi(x)$. Then, we note
\begin{align}
\sum_{x \in \bS} \bw(x) \cdot \| \Pi(x) - c'(x) \|_2^z = \Ex_{\bx \sim \bS}\left[ \frac{1}{\tilde{\sigma}(\bx)} \cdot \| \Pi(\bx) - c'(\bx)\|_2^z \right], \label{eq:estimator-sample-S}
\end{align}
so that in expectation over $\bS$, we have
\begin{align*}
\Ex_{\bS}\left[ \Ex_{\bx \sim \bS}\left[ \frac{1}{\tilde{\sigma}(\bx)} \cdot \| \Pi(\bx) - c'(\bx)\|_2^z \right]\right] = \Ex_{\bx \sim \tilde{\sigma}}\left[\frac{1}{\tilde{\sigma}(\bx)} \cdot \| \Pi(\bx) - c'(\bx)\|_2^z \right] = \cost_{z}^z(\Pi(X), C'). 
\end{align*}
By Markov's inequality, we obtain our desired bound.
\end{proof}

\begin{corollary}\label{cor:bad-approx}
Let $X = \{ x_1, \dots, x_n \} \subset \R^d$ be any set of points, and $C \subset \R^d$ of size $k$ be centers for optimally $(k,z)$-clustering $X$. For any $\eps \in (0,1/2)$, let $\calJ_{d,t}$ be the Johnson-Lindenstrauss map with
\[ t \gsim \dfrac{z \log (1/\eps) + \log k}{\eps^2},\]
Then, with probability at least $0.97$ over the draw of $\bPi \sim \calJ_{d,t}$,
\begin{align}
\dfrac{1}{100^{1/z} (1+\eps)^{1+1/z}} \cdot \cost_z\left(X, C \right) \leq \min_{\substack{C' \subset \R^t \\ |C'| = k}} \cost_{z}(\bPi(X), C'). \label{eq:optimum-not-decreased}
\end{align}
\end{corollary}

\begin{proof}
We sample $\bPi \sim \calJ_{d,t}$ and $(\bS, \bw)$ as per Definition~\ref{def:events}. By Lemma~\ref{lem:jl}, Lemma~\ref{lem:bad-approx}, Theorem~\ref{thm:strong-coresets}, and a union bound, we have events $\bE_1$, $\bE_2$, and $\bE_3(100)$ hold with probability at least $0.97$. Hence, we obtain the desired result from applying Lemma~\ref{lem:combine-events}.
\end{proof}

\subsubsection{Improving the Approximation}\label{sec:improving-approx}

In what follows, we will improve upon the approximation of Corollary~\ref{cor:bad-approx} significantly, to show that with large probability, $\bE_3(1+\eps)$ holds. We let $X = \{ x_1,\dots, x_n \} \subset \R^d$ and denote $C \subset \R^d$ of size $k$ to be the optimal $(k,z)$-clustering of $X$. As before, we let $c \colon X \to C$ map each $x \in X$ to the point center in $C$, and $\sigma \colon X \to \R_{\geq 0}$ be the sensitivities of $X$ with respect to $C$ as in Lemma~\ref{lem:sens}, and $\tilde{\sigma}$ be the sensitivity sampling distribution. 

We define one more event, $\bE_4$, with respect to the randomness of $\bPi \sim \calJ_{d,t}$. First, we let $\bD_x \in \R_{\geq 0}$ denote the random variable given by
\begin{align}
\bD_x \eqdef \dfrac{\|\bPi(x) - \bPi(c(x))\|_2}{\|x - c(x)\|_2}. \label{eq:def-D}
\end{align}
Notice that when $\bPi$ consists of i.i.d $\calN(0, 1/t)$, then $t \bD_x^2$ is distributed as $\chi^2$-random variable with $t$ degrees of freedom. We say event $\bE_4$ holds whenever
\begin{align}
\sum_{x \in X} \bD_x^{2z} \cdot \sigma(x) \leq 100(k+1) 2^z.  \label{eq:event-4}
\end{align}
\begin{claim}\label{cl:event-4}
With probability at least $0.99$ over the draw of $\bPi \sim \calJ_{d,t}$, event $\bE_4$ holds.
\end{claim}

\begin{proof}
The proof will simply follow from computing the expectation of the left-hand side of (\ref{eq:event-4}), and applying Markov's inequality. In particular, for every $x \in X$, we can apply Lemma~\ref{lem:guassian-guarantee} to conclude have
\begin{align*}
\Ex_{\bPi \sim \calJ_{d,t}}\left[ \bD_{x}^{2z}\right] \leq 2^{z}. 
\end{align*}
The remaining part follows from the bound on $\frakS_{\sigma}$.
\end{proof}

\begin{lemma}\label{lem:variance-bound}
Let $\Pi \in \calJ_{d,t}$ be a Johnson-Lindenstrauss map where, for $\alpha > 1$, the following events holds:
\begin{enumerate}
\item\label{en:ub-1} Guarantee from Lemma~\ref{lem:upper-bound}: $\sum_{x \in X} \| \Pi(x) - \Pi(c(x))\|_2^z \leq \alpha \cdot \cost_z^z(X, C)$.
\item\label{en:ub-2} Guarantee from Corollary~\ref{cor:bad-approx}: letting $C' \subset \R^t$ be the optimal $(k,z)$-clustering of $\Pi(X)$, then $\cost_z^z(\Pi(X), C') \geq (1/\alpha) \cdot \cost_z^z(X, C)$.
\item\label{en:ub-3} Event $\bE_4$ holds. 
\end{enumerate}
Then, if we let $(\bS, \bw)$ denote $m = \poly(k,2^z, 1/\eps, \alpha)$ i.i.d draws from $\tilde{\sigma}$ and $\bw(x) = 1/(m \tilde{\sigma}(x))$, with probability at least $0.99$,
\[ \cost_z^z((\Pi(\bS), \bw), C') \leq (1+\eps) \cdot \cost_z^z(\Pi(X), C'). \]
\end{lemma}

\begin{proof}
The proof follows the same schema as Lemma~\ref{lem:bad-approx}. However, we give a bound on the variance of the estimator in order to improve upon the use of Markov's inequality. Specifically, we compute the variance of a rescaling of (\ref{eq:estimator-sample-S}).
\begin{align}
&\mathop{\Var}_{\bS}\left[ \Ex_{\bx \sim \bS}\left[ \frac{1}{\tilde{\sigma}(\bx)} \cdot \dfrac{\|\Pi(\bx) - c'(\bx)\|_2^z}{\cost_z^z(\Pi(X), C')}\right] \right] \leq \frac{1}{m} \Ex_{\bx \sim \tilde{\sigma}}\left[ \left(\frac{1}{\tilde{\sigma}(\bx)} \cdot \dfrac{\| \Pi(\bx) - c'(\bx)\|_2^z}{\cost_z^z(\Pi(X), C')} \right)^2\right]  \nonumber \\
&\qquad\qquad\qquad\qquad= \frac{\frakS_{\sigma}}{m} \sum_{x \in X} \left(\frac{1}{\sigma(x)} \cdot \dfrac{\| \Pi(x) - c'(x)\|_2^{2z}}{\cost_z^{2z}(\Pi(X), C')} \right). \label{eq:var-bound}
\end{align}
By the same argument as in the proof of Lemma~\ref{lem:sens} (applying the triangle inequality twice),
\begin{align}
\dfrac{\| \Pi(x) - c'(x) \|_2^z}{\cost_z^z(\Pi(X), C')} &\leq 2^{z-1} \cdot \dfrac{\| \Pi(x) - \Pi(c(x))\|_2^z}{\cost_z^z(\Pi(X), C')} + \dfrac{2^{2(z-1)}}{|X_{c(x)}|} \left( \dfrac{\cost_z^z(\Pi(X), \Pi(C))}{\cost_z^z(\Pi(X), C')} + 1\right). \label{eq:sensitivity-bound}
\end{align}
Recall that we have the lower bound (\ref{en:ub-2}) and the upper bound (\ref{en:ub-1}). Hence, along with the definition of $D_x$ in (\ref{eq:def-D}) (we remove the bold-face as $\Pi$ is fixed), we upper bound the left-hand side of (\ref{eq:sensitivity-bound}) by
\begin{align*}
&\alpha \cdot 2^{z-1}\cdot \dfrac{D_x^{z} \cdot \| x - c(x) \|_2^z}{\cost_z^z(X, C)} + \dfrac{2^{2(z-1)}}{|X_{c(x)}|} \left( \alpha^2 + 1\right) \leq \alpha^2 \cdot 2^{100z} ( D_x^z + 1) \cdot \sigma(x). 
\end{align*}
In particular, we may plug this in to (\ref{eq:var-bound}) and use the definition of $\sigma(x)$. Specifically, one obtains the variance in (\ref{eq:var-bound}) is upper bounded by
\begin{align*}
\frac{\frakS_{\sigma}\cdot \alpha^4 \cdot 2^{200z}}{m} \sum_{x \in X} (D_x^z + 1)^2 \cdot \sigma(x) &\leq \frac{4\frakS_{\sigma}^2 \cdot \alpha^4 \cdot 2^{200z}}{m} + \frac{4 \frakS_{\sigma}\cdot \alpha^4 \cdot 2^{200z}}{m} \sum_{x \in X} D_x^{2z} \cdot \sigma(x) \\
	&\leq \frac{1000 \frakS_{\sigma}^2 \cdot \alpha^4 \cdot 2^{300 z}}{m},
\end{align*}
where in the final inequality, we used the fact that $\bE_4$ holds. Hence, letting $m$ be a large enough polynomial in $\poly(k, 2^z, 1/\eps, \alpha)$ implies the variance is smaller than $o(\eps^2)$, so we can apply Chebyshev's inequality.
\end{proof}

\begin{corollary}
Let $X = \{ x_1,\dots, x_n \} \subset \R^d$ by any set of points, and $C \subset \R^d$ of size $k$ be centers for optimally $(k,z)$-clustering $X$. For any $\eps \in (0, 1/2)$, let $\calJ_{d,t}$ be the Johnson-Lindenstrauss map with 
\[ t \gsim \frac{z \log(1/\eps) + \log k}{\eps^2}. \]
Then, with probability at least $0.92$ over the draw of $\bPi \sim \calJ_{d,t}$, 
\[ \dfrac{1}{(1+\eps)^{1+2/z}} \cdot \cost_z(X, C) \leq \min_{\substack{C' \subset \R^t \\ |C'| \leq k}} \cost_z(\bPi(X), C'). \]
\end{corollary}

\begin{proof}
We sample $\bPi \sim \calJ_{d,t}$ and $(\bS, \bw)$ as per Definition~\ref{def:events}. Note that by Theorem~\ref{thm:strong-coresets} and the setting of $m$, event $\bE_1$ holds with probability at least $0.99$ over the draw of $(\bS, \bw)$. By Lemma~\ref{lem:upper-bound}, Corollary~\ref{cor:bad-approx}, and Claim~\ref{cl:event-4}, and the setting of $m$ and $t$, the conditions (\ref{en:ub-1}), (\ref{en:ub-2}) and (\ref{en:ub-3}) of Lemma~\ref{lem:variance-bound} hold with probability at least $0.95$, with $\alpha $ being set to a large enough constant, and hence event $\bE_3(1+\eps)$ holds with probability at least $0.94$. Finally, event $\bE_2$ holds with probability $0.99$ by Lemma~\ref{lem:jl}, and taking a union bound and Lemma~\ref{lem:combine-events} gives the desired bound.
\end{proof}

\section{Subspace Approximation}\label{sec:subspace}

In the $(k,z)$-\emph{subspace approximation} problem, we consider a subspace $R \subset \R^d$ of dimension less than $k$, which we may encode by a collection of at most $k$ orthonormal vectors $r_1,\dots, r_k \in R$. We let $\rho_R \colon \R^d \to \R^d$ denote the map which sends each vector $x \in \R^d$ to its closest point in $R$, and note that
\begin{align*}
\rho_R(x) &= \argmin_{z \in R} \| x - z\|_2^2 = \sum_{i=1}^k \langle x, r_i \rangle \cdot r_i \in \R^d,
\end{align*}
For any subset $X \subset \R^d$ and any $k$-dimensional subspace $R$, we let 
\[ \cost_z^z(X, R) = \sum_{x \in X} \left\| x - \rho_R(x) \right\|_2^z. \]
In this section, we will show that we may compute the optimum $k$-subspace approximation after applying a Johnson-Lindenstrauss transform.
\begin{theorem}[Johnson-Lindenstrauss for $k$-Subspace Approximation]\label{thm:k-subspace}
Let $X = \{ x_1,\dots, x_n \} \subset \R^d$ be any set of points, and let $R$ denote the optimal $(k, z)$-subspace approximation of $X$. For any $\eps \in (0, 1)$, suppose we let $\calJ_{d, t}$ be the distribution over Johnson-Lindenstrauss maps where
\[ t \gsim \dfrac{z \cdot k^2 \cdot \polylog(k/\eps)}{\eps^3}. \]
Then, with probability at least $0.9$ over the draw of $\bPi \sim \calJ_{d, t}$,
\begin{align*}
\dfrac{1}{1 + \eps} \cdot \cost_z(X, R) \leq \min_{\substack{R' \subset \R^t \\ \dim R' \leq k}} \cost_z(\bPi(X), R') \leq (1+\eps) \cdot \cost_z(X, R).
\end{align*}
\end{theorem}

The proof of the above theorem proceeds in two steps, and models the argument in the previous section. First, we show that the cost of the optimum does not increase substantially (the right-most inequality in the theorem). This is done in the next subsection. The second step is showing that the optimum does not decrease substantially (the left-most inequality in the theorem). The second step is done in the subsequent subsection.

\subsection{Easy Direction: Optimum Cost Does Not Increase}

The first direction, which shows that the optimal $(k,z)$-subspace approximation does not increase follows similarly to Lemma~\ref{lem:upper-bound}.

\begin{lemma}\label{lem:subspace-ub}
Let $X = \{ x_1,\dots, x_n \} \subset \R^d$ by any set of points and let $R \subset \R^d$ be optimal $(k,z)$-subspace approximation of $X$. For any $\eps \in (0, 1)$, we let $\calJ_{d,t}$ be the distribution over Johnson-Lindenstrauss maps. If $t \gsim z / \eps^2$, then with probability at least $0.99$ over the draw of $\bPi \sim \calJ_{d,t}$, 
\[ \sum_{x \in X} \| \bPi(x) - \bPi(\rho_R(x))\|_2^z \leq (1+\eps) \cdot \cost_z(X, R), \]
and hence,
\[ \min_{\substack{R' \subset \R^t \\ \dim R' \leq k}} \cost_z(\bPi(X), R') \leq (1+\eps) \min_{\substack{R \subset \R^d \\ \dim R \leq k}} \cost_z(X, R). \]
\end{lemma}

\begin{proof}
Note that $\bPi$ is a linear map, so if we let $r_1,\dots, r_k \in R$ denote $k$ orthonormal unit vectors spanning $R$, then $\bPi(r_1), \dots, \bPi(r_k) \in \R^t$ are $k$ vectors spanning the subspace $\bPi(R)$. Furthermore, we may consider the $k$-dimensional subspace 
\[ \bPi(R) \eqdef \left\{ \sum_{i=1}^k \alpha_i \cdot \bPi(r_i) \in \R^t : \alpha_1,\dots, \alpha_k \in \R \right\}. \]
Notice that for any $x \in \R^d$, by linearity of $\bPi$,
\begin{align*}
\bPi(\rho_{R}(x)) &= \sum_{i=1}^k \langle x - \tau, r_i \rangle \cdot  \bPi(r_i) \in \bPi(F),
\end{align*}
which means that we may always upper bound 
\begin{align}
\min_{\substack{R' \subset \R^t\\ \dim R' \leq k}} \cost_z(\bPi(X), R') &\leq \left( \sum_{x \in X} \left\| \bPi(x) - \bPi(\rho_R(x)) \right\|_2^z \right)^{1/z}. \label{eq:specific-F}
\end{align}
It hence remains to upper bound the left-hand side of (\ref{eq:specific-F}). We now use the fact that $\bPi$ is drawn from a Johnson-Lindenstrauss distribution. Specifically, the lemma follows from applying Markov's inequality once we show
\[ \dfrac{1}{\cost_z^z(X, F)} \Ex_{\bPi \sim \calJ_{d,t}}\left[ \sum_{x\in X} \left(\| \bPi(x) - \bPi(\rho_F(x))\|_2^z - \| x - \rho_F(x)\|_2^z \right)^+\right] \leq \dfrac{(1+\eps)^z - 1}{100}, \]
which follows from Lemma~\ref{lem:gaussian-guarantee}.
\end{proof}

\subsection{Hard Direction: Optimum Cost Does Not Decrease}

\subsubsection{Preliminaries}

In the $(k,z)$-subspace approximation problem, there will be a difference between complexities of known strong coresets and weak coresets. Our argument will only use weak coresets, which is important for us, as strong coresets have a dependence on $d$ (which we are trying to avoid).

\begin{definition}[Weak Coresets for $(k,z)$-subspace approximation]
Let $X = \{ x_1,\dots, x_n \} \subset \R^d$ be a set of points. A weak $\eps$-coreset of $X$ for $(k,z)$-subspace approximation is a weighted subset $S \subset \R^d$ of points with weights $w \colon S \to \R_{\geq 0}$ such that, 
\[ \frac{1}{1+\eps} \cdot \min_{\substack{R \subset \R^d \\ \dim R \leq k}} \cost_z(X, R) \leq  \min_{\substack{R \subset \R^d \\ \dim R \leq k}} \cost_z((S, w), R) \leq (1+\eps) \cdot  \min_{\substack{R \subset \R^d \\ \dim R \leq k}} \cost_z(X, R). \]
\end{definition}

Similarly to the case of $(k,z)$-clustering, algorithms for building weak coresets proceed by sampling according to the sensitivity framework. We proceed by defining sensitivity functions in the context of subspace approximation, and then state a lemma which gives a sensitivity function that we will use.
\begin{definition}[Sensitivities]
Let $n, d \in \N$, and consider any set of points $X = \{ x_1,\dots, x_n \} \subset \R^d$, as well as $k \in \N$ and $z \geq 1$. A sensitivity function $\sigma \colon X \to \R_{\geq 0}$ for $(k,z)$-subspace approximation in $\R^d$ is a function satisfying that, for all $x \in X$,
\begin{align*}
\sup_{\substack{R \subset \R^d \\ \dim R \leq k}} \dfrac{\|x - \rho_R(x)\|_2^z}{\cost_z^z(X, R)} \leq \sigma(x). 
\end{align*}
The total sensitivity of the sensitivity function $\sigma$ is given by
\[ \frakS_{\sigma} = \sum_{x \in X} \sigma(x). \]
For a sensitivity function, we let $\tilde{\sigma}$ denote the sensitivity sampling distribution, supported on $X$, which samples $x \in X$ with probability proportional to $\sigma(x)$. 
\end{definition}

We now state a specific sensitivity function that we will use. The proof will closely follow a method for bounding the total sensitivity of \cite{VX12}. The resulting weak $\eps$-coreset will have a worse dependence than the best-known coresets for this problem; however, the specific form of the sensitivity function will be especially useful for us. Specifically, the non-optimality of the sensitivity function will not significantly affect the final bound on dimension reduction.

\begin{lemma}[Theorem 18 of~\cite{VX12}]\label{lem:subspace-sensitivities}
Let $n, d \in \N$, and consider any set of points $X = \{ x_1,\dots, x_n \} \subset \R^d$, as well as $k \in \N$ with $k < d$, and $z \geq 1$. Suppose $R \subset \R^d$ is the optimal $(k,z)$-subspace approximation of $X$ in $\R^d$. Then, the function $\sigma \colon X \to \R_{\geq 0}$ given by 
\begin{align*} 
\sigma(x) = 2^{z-1} \cdot \dfrac{\|x - \rho_R(x)\|_2^z}{\cost_z^z(X, R)} + 2^{2z-1} \cdot \sup_{u \in \R^d} \dfrac{|\langle \rho_R(x), u \rangle|^z}{\sum_{x' \in X} |\langle \rho_R(x'), u \rangle|^z}
\end{align*}
is a sensitivity function for $(k,z)$-subspace approximation of $X$ in $\R^d$, satisfying
\[ \frakS_{\sigma} \leq 2^{z-1} + 2^{2z-1} (k+1)^{1+z}. \]
\end{lemma}

\begin{proof}
Consider any subspace $R' \subset \R^d$ of dimension at most $k$. Then, for any $x \in X$
\begin{align}
&\| x - \rho_{R'}(x) \|_2^z \leq \| x - \rho_{R'}(\rho_{R}(x))\|_2^z \leq 2^{z-1} \left( \| x - \rho_{R}(x) \|_2^z + \| \rho_{R}(x) - \rho_{R'}(\rho_{R}(x))\|_2^z\right) \nonumber \\
	&\qquad\leq 2^{z-1} \left(\dfrac{\| x - \rho_R(x)\|_2^z}{\cost_z^z(X, R)} \cdot \cost_z^z(X, R) + \dfrac{\| \rho_R(x) - \rho_{R'}(\rho_R(x))\|_2^z}{\cost_z^z(\rho_{R}(X), R')} \cdot \cost_z^z(\rho_R(X), R')\right).\label{eq:sub-sens}
\end{align}
Notice that $\cost_z^z(\rho_R(X), R') \leq 2^{z-1}( \cost_z^z(X, R) + \cost_z^z(X, R'))$ by the triangle inequality and H\"{o}lder's inequality, and that $\cost_z^z(X, R) \leq \cost_z^z(X, R')$ since $R$ is the optimal $(k,z)$-subspace approximation. Hence, dividing the left- and right-hand side of (\ref{eq:sub-sens}) we have
\begin{align*}
\dfrac{\|x - \rho_{R'}(x)\|_2^z}{\cost_z^z(X, R')} \leq 2^{z-1} \cdot \dfrac{\|x - \rho_R(x)\|_2^z}{\cost_z^z(X, R)} + 2^{2z- 1} \cdot \dfrac{\|\rho_R(x) - \rho_{R'}(\rho_{R}(x))\|_2^z}{\cost_z^z(\rho_R(X), R')}.
\end{align*}
It remains to show that, for any set of points $Y \subset \R^d$ (in particular, the set $\{ \rho_R(x) : x \in X\}$), and any $y \in Y$,
\begin{align*}
\sup_{\substack{R' \subset \R^d \\ \dim R' \leq k}} \dfrac{\| y - \rho_{R'}(y)\|_2^z}{\cost_z^z(Y, R')} \leq \sup_{\substack{H \subset \R^d \\ \dim H = d-1}} \dfrac{\| y - \rho_H(y)\|_2^z}{\cost_z^z(Y, H)} = \sup_{u \in \R^d} \dfrac{|\langle y, u\rangle|^z}{\sum_{y' \in Y} |\langle y', u\rangle|^z}.
\end{align*}
In particular, note that for any subspace $R' \subset \R^d$ of dimension at most $k$, there exists a $(d-1)$-dimensional subspace $H \subset \R^d$ containing $R'$ given by all vectors orthogonal to $y- \rho_{R'}(y)$. In particular, $\cost_z^z(Y, H) \leq \cost_z^z(Y, R')$ since $R'$ is contained in $H$, and $\| y - \rho_H(y)\|_2= \|y - \rho_{R'}(y)\|_2$ by the definition of $H$. The bound on the total sensitivity then follows from Lemma~16 in~\cite{VX12}, where we use the fact that $\{ \rho_R(x) : x \in X \}$ lies in a $k$-dimensional subspace.
\end{proof}

We will use the following geometric theorem of \cite{SV12} in our proof. The theorem says that an approximately optimal $(k,z)$-subspace approximation lies in the span of a small set of points. We state the lemma for weighted point sets, even though \cite{SV12} state it for unweighted points. We note that adding weights can be simulated by replicating points. 

\begin{lemma}[Theorem 3.1~\cite{SV12}]\label{lem:find-good-subspace}
Let $d, k \in \N$, and consider a weighted set of points $S \subset \R^d$ with weights $w \colon S \to \R_{\geq 0}$, as well as $\eps \in (0,1)$ and $z \geq 1$. There exists a subset $Q \subset S$ of $O(k^2\log(k/\eps)/\eps)$ and a $k$-dimensional subspace $R' \subset \R^d$ within the span of $Q$ satisfying
\[ \cost_z((S, w), R') \leq (1+\eps) \min_{\substack{R \subset \R^d \\ \dim R \leq k}} \cost_z((S, w), R). \]
\end{lemma}

Lemma~\ref{lem:subspace-sensitivities} gives us an appropriate sensitivity function, and Lemma~\ref{lem:find-good-subspace} limits the search of the subspace to just a small set of points. Similarly to the case of $(k,z)$-clustering, we can use this to construct weak $\eps$-coresets for $(k,z)$-subspace approximation. The following theorem is Theorem~5.10 from \cite{HV20}. We state the theorem with the sensitivity function of Lemma~\ref{lem:subspace-sensitivities}.

\begin{theorem}[Theorem~5.10 from \cite{HV20}]\label{thm:weak-coresets-sa}
For any subset $X = \{ x_1,\dots, x_n \} \subset \R^d$ and $\eps \in (0, 1/2)$, let $\tilde{\sigma}$ denote the sensitivity sampling distribution from the sensitivity function of Lemma~\ref{lem:subspace-sensitivities}. 
\begin{itemize}
\item Let $(\bS, \bw)$ denote the random (multi-)set $\bS \subset X$ and $\bw \colon \bS \to \R_{\geq 0}$ given by, for 
\[ m = \poly((k+1)^{z}, 1/\eps) \]
iterations, sampling $\bx \sim \tilde{\sigma}$ i.i.d and letting $\bw(\bx) = 1/(m\tilde{\sigma}(x))$.
\item Then, with probability $1 - o(1)$ over the draw of $(\bS, \bw)$, it is an $\eps$-weak coreset for $(k,z)$-subspace approximation of $X$.
\end{itemize}
\end{theorem}

\subsection{The Important Events} 

Similarly to the previous section, we define the important events, over the randomness in $\bPi$ such that, if these are satisfied, then the optimum of $(k,z)$-subspace approximation after dimension reduction does not decrease substantially. We first define the events, and then we prove that if the events are all satisfied, then we obtain our desired approximation. 

\begin{definition}[The Events]\label{def:events-sa} Let $X = \{ x_1,\dots, x_n \} \subset \R^d$, and $\tilde{\sigma}$ the sensitivity sampling distribution of $X$ from Lemma~\ref{lem:subspace-sensitivities}. We consider the following experiment,
\begin{enumerate}
\item We generate a sample $(\bS, \bw)$ by sampling from $\tilde{\sigma}$ for $m = \poly(k^z, 1/\eps)$ i.i.d iterations $\bx \sim \tilde{\sigma}$ and set $\bw(\bx) = 1/(m\tilde{\sigma}(\bx))$.
\item Furthermore, we sample $\bPi \sim \calJ_{d, t}$, which is a Johnson-Lindenstrauss map $\R^d \to \R^t$.
\item We let $\bS' = \bPi(\bS) \subset \R^t$ denote the image of $\bPi$ on $\bS$.
\end{enumerate}
The events are the following:
\begin{itemize}
\item $\bE_1$ : The weighted (multi-)set $(\bS, \bw)$ is a weak $\eps$-coreset for $(k,z)$-subspace approximation of $X$ in $\R^d$.
\item $\bE_2$ : The map $\bPi \colon \R^d \to \R^t$ satisfies the following condition. For any choice of $O(k^2 \log(k/\eps)/\eps)$ points of $\bS$, $\bPi$ is an $\eps$-subspace embedding of the subspace spanned by these points. 
\item $\bE_3(\beta)$ : Let $\bR' \subset \R^t$ denote the $k$-dimensional subspace for optimal $(k,z)$-subspace approximation of $\bPi(X)$ in $\R^t$. Then, 
\[ \cost_z^z((\bPi(\bS), \bw), \bR') \leq \beta \cdot \cost_z^z(\bPi(X), \bR'). \]
\end{itemize}
\end{definition}

\begin{lemma}\label{lem:events-to-thm-sa}
Let $X = \{ x_1,\dots, x_n \} \subset \R^d$, and suppose $(\bS, \bw)$ and $\bPi \colon \R^d \to \R^t$ satisfy events $\bE_1$, $\bE_2$, and $\bE_3(\beta)$. Then,
\[ \min_{\substack{R' \subset \R^t \\ \dim R' \leq k}} \cost_z(\bPi(X), R') \geq \dfrac{1}{\beta^{1/z}(1+\eps)^3} \cdot \min_{\substack{R \subset \R^d \\ \dim R \leq k}} \cost_z(X, R). \]
\end{lemma}

\begin{proof}
Consider a fixed $\bPi$ and $(\bS, \bw)$ satisfying the three events of Definition~\ref{def:events-sa}. Let $\bR' \subset \R^t$ be the $k$-dimensional subspace which minimizes $\cost_z^z(\bPi(X), \bR')$. Then, by event $\bE_3(\beta)$, we have $\cost_z^z((\bPi(\bS), \bw), \bR') \leq \beta \cdot \cost_z^z(\bPi(X), \bR')$. Now, we apply Lemma~\ref{lem:find-good-subspace} to $(\bPi(\bS), \bw)$, and we obtain a subset $Q \subset \bS$ of size $O(k^2 \log(k/\eps) / \eps)$ for which there exists a $k$-dimensional subspace $R'' \subset \R^t$ within the span of $\bPi(Q)$ which satisfies
\[ \left( \sum_{x \in \bS} \bw(x) \cdot \| \bPi(x) - \rho_{R''}(\bPi(x)) \|_2^z\right)^{1/z} = \cost_z((\bPi(\bS),\bw), R'') \leq (1+\eps) \cdot \cost_z((\bPi(\bS), \bw), \bR'). \]
Note that $R''$ is a $k$-dimensional subspace lying in the span of $\bPi(Q)$. For any $x \in \bS$, we will use the fact that $\bE_2$ is satisfied to say that $\bPi$ is an $\eps$-subspace embedding of the subspace spanned by $Q \cup \{x\}$. This will enable us to find a subspace $U \subset \R^d$ of dimension $k$ whose cost of approximating $(\bS, \bw)$ is at most $(1+\eps) \cdot \cost_z( (\bPi(\bS), \bw), R'')$.

Specifically, we write $v_1, \dots, v_k \in \R^t$, as orthogonal unit vectors which span $R''$. Because $R''$ lies in the span of $\bPi(Q)$, there are vectors $u_1,\dots, u_k \in \R^d$ in the span of $Q$ which satisfy
\[ v_{\ell} = \bPi(u_{\ell}) \in \R^t \qquad\text{for}\qquad u_{\ell} = \sum_{y \in Q} c_{\ell, y} y \in \R^d. \]
Hence, the subspace $U$ given by the span of all vectors in $U$ is a $k$-dimensional subspace lying in the span of $Q$. For $x \in \bS$, we may write the coefficients $\gamma_{\ell}(x) = \langle \bPi(x), v_{\ell}\rangle$, and we may express projection $\rho_{R''}(\bPi(x)) \in \R^t$ as
\[ \rho_{R''}(\bPi(x)) = \sum_{\ell=1}^k \gamma_{\ell}(x) \cdot v_{\ell} = \bPi\left( \sum_{y \in Q} \left(\sum_{\ell=1}^k \gamma_{\ell}(x) c_{\ell, y}\right) \cdot y\right) = \bPi\left( \sum_{\ell=1}^k \gamma_{\ell}(x) u_{\ell} \right). \]
which is a linear combination of $Q$. By event $\bE_2$, $\bPi$ is an $\eps$-subspace embedding of the subspace spanned by $Q \cup \{ x\}$, so
\[ \| x - \rho_U(x) \|_2 \leq \left\| x - \sum_{\ell=1}^k \gamma_{\ell}(x) u_{\ell}\right\|_2 \leq (1+\eps) \| \bPi(x) - \rho_{R''}(\bPi(x)) \|_2. \]
Combining the inequalities, we have
\[ \cost_z\left( (\bS, \bw), U\right) \leq (1+\eps) \cdot \cost_z\left( (\bPi(\bS), \bw), R''\right),  \]
and finally, since $(\bS, \bw)$ is a $\eps$-weak coreset, we obtain the desired inequality.
\end{proof}

We note that event $\bE_1$ will be satisfied with sufficiently high probability from Theorem~\ref{thm:weak-coresets-sa}. Furthermore, event $\bE_2$ is satisfied with sufficiently high probability from the following simple lemma. All that will remain is showing that event $\bE_3(\beta)$ is satisfied.

\begin{lemma}\label{lem:event-2-sa}
Let $S \subset \R^d$ be any set of $m$ points and $\ell \in \N$, and let $\calJ_{d,t}$ denote the Johnson-Lindenstrauss map, with
\[t \gsim \dfrac{\ell \log m}{\eps^2}. \]
Then, with probability $0.99$ over the draw of $\bPi \sim \calJ_{d, t}$, $\bPi$ is an $\eps$-subspace embedding for all subspaces spanned by $\ell$ vectors of $S$.
\end{lemma}

\begin{proof}
There are at most $\binom{m}{\ell}$ subspaces spanned by $\ell$ vectors of $S$. If $\bPi$ is a subspace embedding for all of them, we obtain our desired conclusion. We use Lemma~\ref{lem:subspace-embedding} with $\delta$ to be a sufficiently small constant factor of $1/m^{\ell}$ and union bound.
\end{proof}

\subsubsection{A Bad Approximation Guarantee}

\begin{lemma}[Warm-Up Lemma]\label{lem:bad-approx-sa}
Fix any $\Pi \in \calJ_{d,t}$ and let $R' \subset \R^t$ denote the $k$-dimensional subspace for optimal $(k,z)$-subspace approximation of $\Pi(X)$ in $\R^t$. Then, with probability at least $0.99$ over the draw of $(\bS, \bw)$ as per Definition~\ref{def:events-sa},
\begin{align*}
\sum_{x \in \bS} \bw(x) \cdot \| \Pi(x) - \rho_{R'}(\Pi(x))\|_2^z \leq 100 \cdot \cost_z^z(\Pi(X), R').
\end{align*}
In other words, event $\bE_3(100)$ holds with probability at least $0.99$.
\end{lemma}

\begin{proof}
Similarly to the proof of Lemma~\ref{lem:bad-approx}, we compute the expectation of the left-hand side of the inequality and use Markov's inequality.
\begin{align}
\Ex_{\bS}\left[ \Ex_{\bx \sim \bS}\left[\frac{1}{\tilde{\sigma}(\bx)} \cdot \| \Pi(x) - \rho_{R'}(\Pi(x))\|_2^z \right]\right] &= \Ex_{\bx \sim \tilde{\sigma}}\left[ \frac{1}{\tilde{\sigma}(\bx)} \cdot \| \Pi(x) - \rho_{R'}(\Pi(x))\|_2^z\right] = \cost_z^z(\Pi(X), R'). \label{eq:expectation-S-sa}
\end{align}
\end{proof}

\begin{corollary}\label{cor:bad-approx-sa}
Let $X = \{ x_1,\dots, x_n \} \subset \R^d$ be any set of points. For any $\eps \in (0,1/2)$, let $\calJ_{d,t}$ be the Johnson-Lindenstrauss map with
\[ t \gsim \dfrac{z \cdot k^2 \cdot \polylog(k/\eps)}{\eps^3}. \]
Then, with probability at least $0.97$ over the draw of $\bPi \sim \calJ_{d,t}$,
\begin{align*}
\dfrac{1}{100^{1/z} (1+\eps)^3} \cdot \cost_z(X, R) \leq \min_{\substack{R' \subset \R^t \\ \dim R' \leq k}} \cost_z(\bPi(X), R').
\end{align*}
\end{corollary}

\begin{proof}
We sample $\bPi \sim \calJ_{d,t}$ and $(\bS, \bw)$ as per Definition~\ref{def:events-sa}. By Theorem~\ref{thm:weak-coresets-sa} and Lemma~\ref{lem:bad-approx-sa}, and a union bound, events $\bE_1$ and $\bE_3(100)$ hold with probability at least $0.98$. Event $\bE_2$ occurs with probability at least $0.99$ by apply Lemma~\ref{lem:event-2-sa} with $m = \poly(k^z, 1/\eps)$ and $\ell = O(k^2 \log(k/\eps) / \eps)$. Hence, we apply Lemma~\ref{lem:events-to-thm-sa}.
\end{proof}

\subsubsection{Improving the Approximation}

We now improve on the approximation of Corollary~\ref{cor:bad-approx-sa} in a fashion similar to that of Subsection~\ref{sec:improving-approx}. We will show that with large probability, $\bE_3(1+\eps)$ holds. We let $X = \{ x_1,\dots, x_n \} \subset \R^d$ and $R \subset \R^d$ be the subspace of dimension $k$ for optimal $(k,z)$-subspace approximation of $X$ in $\R^d$. As before, we let $\sigma \colon X \to \R_{\geq 0}$ be the sensitivities of $X$ with respect to $R$ (as in Lemma~\ref{lem:subspace-sensitivities}), and $\tilde{\sigma}$ be the sensitivity sampling distribution. 

We define one more events, $\bE_4$ with respect to the randomness in $\bPi \sim \calJ_{d,t}$. Let $\bD_x \in \R_{\geq 0}$ denote the random variable given by
\begin{align} 
\bD_x \eqdef \dfrac{\| \bPi(x) - \bPi(\rho_R(x))\|_2}{\|x - \rho_R(x)\|_2}. \label{eq:def-D-sa}
\end{align}
We say event $\bE_4$ holds whenever
\begin{align} 
\sum_{x \in X} \bD_x^{2z} \cdot \sigma(x) \leq 100 \cdot 2^z \cdot \frakS_{\sigma}, \label{eq:event-4-sa}
\end{align}
which holds with probability at least $0.99$ over the draw of $\bPi \sim \calJ_{d,t}$, similarly to the proof of Claim~\ref{cl:event-4} and Lemma~\ref{lem:subspace-sensitivities}. 

\begin{lemma}\label{lem:variance-bound-sa}
Let $\Pi \in \calJ_{d,t}$ be a Johnson-Lindenstrauss map where, for $\alpha > 1$, the follows events hold:
\begin{enumerate}
\item\label{en:event-1-sa} Guarantee from Lemma~\ref{lem:subspace-ub}: $\sum_{x \in X} \| \Pi(x) - \Pi(\rho_R(x))\|_2^z \leq \alpha \cdot \cost_z^z(X, R)$.
\item\label{en:event-2-sa} Guarantee from Corollary~\ref{cor:bad-approx-sa}: letting $R' \subset \R^t$ be the optimal $(k,z)$-subspace approximation of $\Pi(X)$, then $\cost_z^z(X, R) \leq \alpha \cost_z^z(\Pi(X), R')$.
\item\label{en:event-3-sa} Event $\bE_4$ holds.
\end{enumerate}
Then, if we let $(\bS, \bw)$ denote $m = \poly(k^z, 1/\eps, \alpha)$ i.i.d draws from $\tilde{\sigma}$ and $\bw(x) = 1 / (m\tilde{\sigma}(x))$, with probability at least $0.99$, 
\[ \cost_z^z((\Pi(\bS), \bw), R') \leq (1+\eps) \cdot \cost_z^z(\Pi(X), R'). \]
\end{lemma}

\begin{proof}
Again, the proof is similar to that of Lemma~\ref{lem:variance-bound}, where we bound the variance of the estimator to apply Chebyshev's inequality. In particular, we have
\begin{align}
\mathop{\Var}_{\bS}\left[ \Ex_{\bx \sim \bS}\left[ \dfrac{1}{\tilde{\sigma}(\bx)} \cdot \dfrac{\| \Pi(\bx) - \rho_{R'}(\Pi(\bx))\|_2^z}{\cost_z^z(\Pi(X), R')} \right] \right] \leq \frac{\frakS_{\sigma}}{m} \sum_{x \in X} \left(\frac{1}{\sigma(x)} \cdot \dfrac{\| \Pi(x) - \rho_{R'}(\Pi(x))\|_2^{2z}}{\cost_z^{2z}(\Pi(X), R')} \right). \label{eq:variance-bound-sa}
\end{align}
Similarly to the proof of Lemma~\ref{lem:variance-bound}, we will upper bound
\begin{align*}
\dfrac{\| \Pi(x) - \rho_{R'}(\Pi(x))\|_2^z}{\cost_z^z(\Pi(X), R)}
\end{align*}
as a function of $\sigma(x)$ and $D_x$ (given by (\ref{eq:def-D-sa}) without boldface as $\Pi$ is fixed). Toward this bound, we simplify the notation by letting $y_x = \rho_{R}(x) \in \R^d$ and $Y = \{ y_x : x \in X\}$. Then, since $\Pi \colon \R^d \to \R^t$ is a linear map, for any $x \in X$
\begin{align}
\sup_{v \in \R^t} \dfrac{|\langle \Pi(y_x), v \rangle|^z}{\sum_{x' \in X} |\langle \Pi(y_{x'}), v\rangle|^z} \leq \sup_{u \in \R^d} \dfrac{|\langle y_x, u \rangle|^z}{ \sum_{x' \in X} |\langle y_{x'}, u\rangle|^z}.  \label{eq:sups-better}
\end{align}
In particular, if we let $M \in \R^{n\times d}$ be the matrix given by having the rows be points $y_x \in Y$, then writing $\Pi \in \R^{d \times t}$, we have $M \Pi \in \R^{n\times t}$ is the matrix whose rows are $\Pi(y_x)$; in particular, one may compare the left- and right-hand sides of (\ref{eq:sups-better}) by letting $u = \Pi v \in \R^d$. Thus, we have
\begin{align}
&\| \Pi(x) - \rho_{R'}(\Pi(x))\|_2^z \leq \| \Pi(x) - \rho_{R'}(\Pi(y_x))\|_2^z \leq 2^{z-1} \left( \| \Pi(x) - \Pi(y_x)\|_2^z + \| \Pi(y_x) - \rho_{R'}(\Pi(y_x))\|_2^z \right) \nonumber \\
	&\qquad\leq 2^{z-1} \left(D_x^z \cdot \dfrac{\| x - y_x\|_2^z}{\cost_z^z(X, R)} \cdot \cost_z^z(X, R) + \dfrac{\| \Pi(y_x) - \rho_{R'}(\Pi(y_x))\|_2^z}{\cost_z^z(\Pi(Y), R')} \cdot \cost_z^z(\Pi(Y), R') \right). \label{eq:distance-in-projection}
\end{align}
We may now apply the triangle inequality, as well as (\ref{en:event-1-sa}) and (\ref{en:event-2-sa}), and we have
\begin{align}
 \cost_z^z(\Pi(Y), R') &\leq 2^{z-1}\left( \cost_z^z(\Pi(X), R') + \sum_{x\in X} \| \Pi(x) - \Pi(y_x)\|_2^z \right) \nonumber \\
 		&\leq 2^{z-1}\left( \cost_z^z(\Pi(X), R') + \alpha \cdot \cost_z^z(X, R)\right)  \leq 2^{z-1} (1 + \alpha^2) \cdot \cost_z^z(\Pi(X), R'). \label{eq:cost-triangle-ineq}
 \end{align}
Finally, we note that, similarly to the proof of Lemma~\ref{lem:subspace-sensitivities}, 
\begin{align} 
\frac{\| \Pi(y_x) - \rho_{R'}(\Pi(y_x))\|_2^z}{\cost_z^z(\Pi(Y), R')} \leq \sup_{v \in \R^t} \dfrac{|\langle \Pi(y_x), v \rangle|^z}{\sum_{x' \in X} |\langle \Pi(y_{x'}), v \rangle|^z}.  \label{eq:sens-on-y}
\end{align}
Continuing to upper-bound the left-hand side of (\ref{eq:distance-in-projection}) by plugging in (\ref{eq:sups-better}), (\ref{eq:sens-on-y}) and (\ref{eq:cost-triangle-ineq}),
\begin{align*}
\dfrac{\| \Pi(x) - \rho_{R'}(\Pi(x))\|_2^z}{\cost_z^z(\Pi(X), R')} &\leq 2^{z-1} \left(D_x^z \alpha \cdot \dfrac{\| x - y_x\|_2^z}{\cost_z^z(X, Y)} + 2^{z-1}(\alpha^2 + 1) \sup_{u \in \R^d} \dfrac{|\langle y_x, u \rangle|^z}{\sum_{x' \in X} |\langle y_{x'}, u\rangle|^z} \right) \\
	&\leq D_x^z \cdot (\alpha^2 + 1) \cdot \sigma(x). 
\end{align*}
In particular, the bound on the variance in (\ref{eq:variance-bound-sa}) is at most
\[ \dfrac{\frakS_{\sigma}}{m} \cdot 2^{4z} (\alpha^2 + 1)^2 \sum_{x \in X}  D_x^{2z} \sigma(x) \leq \dfrac{\frakS_{\sigma}^2}{m} \cdot 2^{5z} (\alpha^2 + 1),\]
so letting $m = \poly(k^z, 1/\eps, \alpha)$ gives the desired bound on the variance.
\end{proof}

\begin{corollary}
Let $X = \{ x_1,\dots, x_n \} \subset \R^d$ be any set of points, and let $R \subset \R^d$ be the subspace for optimal $(k,z)$-subspace approximation of $X$. For any $\eps \in (0, 1/2)$, let $\calJ_{d,t}$ be the Johnson-Lindenstrauss map with 
\[ t \gsim \dfrac{z \cdot k^2 \cdot \polylog(k/\eps)}{\eps^3}. \]
Then, with probability at least $0.92$ over the draw of $\bPi \sim \calJ_{d,t}$,
\[ \dfrac{1}{(1+\eps)^{3+1/z}} \cdot \cost_z(X, R) \leq \min_{\substack{R' \subset \R^t \\ \dim R' \leq k}} \cost_z(\bPi(X), R'). \]
\end{corollary}

\begin{proof}
We sample $\bPi \sim \calJ_{d,t}$ and $(\bS, \bw)$ as per Definition~\ref{def:events-sa}. Again, Theorem~\ref{thm:weak-coresets-sa} guarantees that $\bE_1$ occurs with probability at least $0.99$ over the draw of $(\bS, \bw)$. By Lemma~\ref{lem:subspace-ub}, Corollary~\ref{cor:bad-approx-sa} and (\ref{eq:event-4-sa}), the condition (\ref{en:event-1-sa}), (\ref{en:event-2-sa}), and (\ref{en:event-3-sa}) are satisfied with probability at least $0.95$, so we may apply Lemma~\ref{lem:variance-bound-sa} and have $\bE_3(1+\eps)$ holds with probability at least $0.94$. Finally, event $\bE_2$ holds with probability at least $0.99$ by Lemma~\ref{lem:event-2-sa}. Taking a union bound and applying Lemma~\ref{lem:events-to-thm-sa} gives the desired bound.
\end{proof}

\section{$k$-Flat Approximation}\label{sec:flat}

In the $(k,z)$-\emph{flat approximation} problem, we consider subspace $R \subset \R^d$ of dimension less than $k$, which we may encode by a collection of at most $k$ orthonormal vectors $r_1,\dots, r_k \in R$, as well as a translation vector $\tau \in \R^d$. The $k$-flat specified by $R$ and $\tau$ is given by the affine subspace
\[ F = \left\{ x + \tau \in \R^d : x \in R \right\}.\]
We let $\rho_F \colon \R^d \to \R$ denote the map which sends each $x \in \R^d$ to its closest point on $F$, and we note that
\begin{align*}
\rho_F(x) &= \argmin_{y \in F} \| x - y\|_2^2 = \tau + \sum_{i=1}^r \langle x - \tau, r_i \rangle r_i
\end{align*}
For any $X \subset \R^d$, we let 
\[ \cost_z^z(X, F) \eqdef \sum_{x \in X} \left\| x - \rho_F(x) \right\|_2^z. \]
In this section, we show that we may find the optimal $k$-flat approximation after applying a Johnson-Lindenstrauss map. The proof will be almost exactly the same as the $(k,z)$-subspace approximation problem. Indeed, it only remains to incorporate a translation vector.
\begin{theorem}[Johnson-Lindenstrauss for $k$-Flat Approximation]\label{thm:k-flat}
Let $X = \{ x_1,\dots, x_n \} \subset \R^d$ be any set of points, and let $F \subset \R^d$ denote the optimal $(k,z)$-flat approximation of $X$. For any $\eps \in (0, 1)$, suppose we let $\calJ_{d,t}$ be a distribution over a Johnson-Lindenstrauss maps where
\[ t \gsim \dfrac{z \cdot k^2 \cdot \polylog(k/\eps)}{\eps^3}. \]
Then, with probability at least $0.9$ over the draw of $\bPi \sim \calJ_{d,t}$,
\[ \dfrac{1}{1+\eps} \cdot \cost_z(X, F) \leq \min_{\substack{F' \text{ $k$-flat} \\ \text{in $\R^t$}}} \cost_z(\bPi(X), F') \leq (1+\eps) \cdot \cost_z(X, F). \]
\end{theorem}

\subsection{Easy Direction: Optimum Cost Does Not Increase}

\begin{lemma}\label{lem:flat-ub}
Let $X = \{ x_1,\dots, x_n \} \subset \R^d$ by any set of points and let $F \subset \R^d$ be the optimal $(k,z)$-flat approximation of $X$. We let $\calJ_{d,t}$ be the distribution over Johnson-Lindenstrauss maps. If $t \gsim z / \eps^2$, then with probability at least $0.99$ over the draw of $\bPi \sim \calJ_{d,t}$, 
\[ \sum_{x \in X} \| \bPi(x) - \bPi(\rho_F(x))\|_2^z \leq (1+\eps) \cdot \cost_z^z(X, F), \]
and hence,
\[ \min_{\substack{F' \text{ $k$-flat} \\ \text{in $\R^t$}}} \cost_z(\bPi(X), F') \leq (1+\eps) \cdot \cost_z(X, F). \]
\end{lemma}

The proof follows in a similar fashion to Lemma~\ref{lem:upper-bound} and Lemma~\ref{lem:subspace-ub}. In particular, there is a natural definition of a $k$-flat $\bPi(F) \subset \R^t$, and the proof proceeds by upper bounding the expected dilation of $\|\bPi(x) - \bPi(\rho_{F}(x))\|_2^z$.


\ignore{\begin{proof}
Let $R \subset \R^d$ be the $k$-dimensional subspace and $\tau \in \R^d$ the translation vector for the $k$-flat $F$. We let $r_1,\dots, r_k \in R$ be $k$ orthonormal vectors spanning $R$. Furthermore, we denote $\bPi(F)$ the $k$-flat in $\R^t$ given by
\[ \bPi(F) \eqdef \left\{ \bPi(\tau) + \sum_{i=1}^r \alpha_i \cdot \bPi(r_i) : \alpha_1,\dots, \alpha_r \in \R \right\}.\]
Since $\bPi \colon \R^d \to \R^t$ is a linear map, we note that 
\[ \bPi(\rho_F(x)) = \bPi(\tau) + \sum_{i=1}^r \langle x - \tau , r_i \rangle \cdot \bPi(r_i) \in \bPi(F),\]
so it suffices to upper bound
\[ \sum_{x \in X} \| \bPi(x) - \bPi(\rho_F(x)) \|_2^z. \]
Similarly to Lemma~\ref{lem:subspace-ub} and Lemma~\ref{lem:upper-bound}, the lemma follows from applying Markov's inequality once we show
\[ \dfrac{1}{\cost_z^z(X, R)} \Ex_{\bPi \sim \calJ_{d,t}}\left[ \sum_{x\in X} \left(\| \bPi(x) - \bPi(\rho_R(x))\|_2^z - \| x - \rho_R(x)\|_2^z \right)^+\right] \lsim \dfrac{(1+\eps)^z - 1}{100}, \]
however, this is also follows (as in Lemma~\ref{lem:upper-bound} and Lemma~\ref{lem:subspace-ub}) from Lemma~\ref{lem:gaussian}.
\end{proof}
}
\subsection{Hard Direction: Optimum Cost Does Not Decrease}

\subsubsection{Preliminaries}

The proof in this section follows similarly to that of $(k,z)$-subspace approximation. 

\begin{definition}[Weak Coresets for $(k,z)$-flat approximation]
Let $X = \{ x_1,\dots, x_n \} \subset \R^d$ be a set of points. A weak $\eps$-coreset of $X$ for $(k,z)$-flat approximation is a weighted subset $S \subset \R^d$ of points with weights $w \colon S \to \R_{\geq 0}$ such that, 
\[ \dfrac{1}{1+\eps} \cdot \min_{\substack{F \text{ $k$-flat} \\ \text{in $\R^d$}}} \cost_z(X, F) \leq \min_{\substack{F \text{ $k$-flat} \\ \text{in $\R^d$}}} \cost_z((S, w), F) \leq (1+\eps) \cdot  \min_{\substack{F \text{ $k$-flat} \\ \text{in $\R^d$}}} \cost_z(X, F) \]
\end{definition}

\begin{definition}[Sensitivities]
Let $n, d \in \N$, and consider any set of points $X = \{ x_1,\dots, x_n \} \subset \R^d$, as well as $k \in \N$ and $z \geq 1$. A sensitivity function $\sigma \colon X \to \R_{\geq 0}$ for $(k,z)$-flat approximation in $\R^d$ is a function satisfying that, for all $x \in X$,
\begin{align*}
\sup_{\substack{F \subset \R^d \\ \text{$k$-flat}}} \dfrac{\|x - \rho_F(x)\|_2^z}{\cost_z^z(X, F)} \leq \sigma(x). 
\end{align*}
The total sensitivity of the senstivity function $\sigma$ is given by
\[\frakS_{\sigma} = \sum_{x \in X} \sigma(x). \]
For a sensitivity function, we let $\tilde{\sigma}$ denote the sensitivity sampling distribution, supported on $X$, which samples $x \in X$ with probability proportional to $\sigma(x)$.
\end{definition}

The sensitivity function we use here generalizes that of the previous section. In particular, the proof will follow similarly, and we will defer to the arguments in the previous section.

\begin{lemma}\label{lem:flat-sensitivities}
Let $n, d \in \N$, and consider any set of points $X = \{ x_1,\dots, x_n \} \subset \R^d$, as well as $k \in \N$ with $k < d$ and $z \geq 1$. Suppose $F \subset \R^d$ is the optimal $(k,z)$-flat approximation of $X$ in $\R^d$. Then, the function $\sigma \colon X \to \R_{\geq 0}$ given by
\[ \sigma(x) = 2^{z-1} \cdot \dfrac{\|x - \rho_F(x)\|_2^z}{\cost_z^z(X, F)} + 2^{2z-1} \cdot \sup_{\substack{u \in \R^d \\ \phi \in \R}} \dfrac{|\langle \rho_F(x), u \rangle - \phi |^z}{\sum_{x' \in X} |\langle \rho_F(x'), u \rangle - \phi|^z} \]
is a sensitivity function for $(k,z)$-flat approximation of $X$ in $\R^d$, satisfying
\[ \frakS_{\sigma} \leq 2^{z-1} + 2^{2z-1} (k+2)^{1+z}.\]
\end{lemma}

\begin{proof}
Consider any $k$-flat $F' \subset \R^d$, given by a subspace $R \subset \R^d$ of dimension at most $k$, and a translation $\tau \in \R^d$. As in the proof of Lemma~\ref{lem:subspace-sensitivities}, 
\begin{align*}
\dfrac{\|x - \rho_{F'}(x)\|_2^z}{\cost_z^z(X, F')} \leq 2^{z-1} \cdot \dfrac{\|x - \rho_{F}(x)\|_2^z}{\cost_z^z(X, F)} + 2^{2z-1}\cdot \dfrac{\|\rho_F(x) - \rho_{F'}(\rho_F(x))\|_2^z}{\cost_z^z(\rho_F(X), F')}.
\end{align*}
We now have that for any $Y \subset \R^d$, and any $y \in Y$, 
\begin{align*}
\sup_{\substack{F' \subset \R^d \\ \text{$k$-flat}}} \dfrac{\|y - \rho_{F'}(y)\|_2^z}{\cost_z^z(Y, F')} \leq \sup_{\tau \in \R^d} \sup_{u \in \R^d} \dfrac{|\langle y - \tau, u \rangle|^z}{\sum_{y' \in Y} |\langle y' - \tau, u \rangle|^z}.
\end{align*}
Finally, for each $y \in Y \subset \R^d$, we may consider appending an additional coordinate and consider $y^* \in \R^{d+1}$ where the $d+1$-th entry is $1$. Then, by linearity
\[ \sup_{\tau \in \R^d} \sup_{u \in \R^d} \dfrac{|\langle y - \tau, u\rangle|^z}{\sum_{y' \in Y}|\langle y-\tau, u\rangle|^z} = \sup_{v \in \R^{d+1}} \dfrac{|\langle y^*, v\rangle|^z}{\sum_{y' \in Y} |\langle y^{'*}, v \rangle|^z},\]
and the bound on the total sensitivity follows from Lemma~\ref{lem:subspace-sensitivities}.
\end{proof}

In the $(k,z)$-subspace approximation section, we used a lemma (Lemma~\ref{lem:find-good-subspace}) to narrow down the approximately optimal subspaces to those spanned by at most $O(k^2 \log(k/\eps) / \eps)$ points. Here, we use a similar lemma in order to find an approximately optimal translation vector $\tau \in \R^d$, which is spanned by a small subset of points.

\begin{lemma}[Lemma 3.3~\cite{SV07}]\label{lem:find-good-translation}
Let $d, k \in \N$, and consider a weighted set of points $S \subset \R^d$ with weights $w \colon S \to \R_{\geq 0}$, as well as $\eps \in (0,1)$ and $z \geq 1$. Suppose $F \subset \R^d$ is the optimal $(k,z)$-flat approximation of $X$, encoded by a $k$-dimensional subspace $R \subset \R^d$ and translation vector $\tau \in \R^d$. There exists a subset $Q \subset S$ of size $O(\log(1/\eps)/\eps)$ and a point $\tau' \in \conv(Q)$ such that the $k$-flat 
\[ F' = \left\{ \tau' + y \in \R^d : y \in R \right\} \]
satisfies
\[ \cost_z((S, w), F') \leq (1+\eps) \cdot \cost_z((S, w), F). \]
\end{lemma}

\begin{theorem}[$\eps$-Weak Coresets for $k$-Flats via Sensitivity Sampling]\label{thm:weak-coresets-flat}
For any subset $X = \{ x_1,\dots, x_n \} \subset \R^d$ and $\eps \in (0, 1/2)$, let $\tilde{\sigma}$ denote the sensitivity sampling distribution. 
\begin{itemize}
\item Let $(\bS, \bw)$ denote the random (multi-)set $\bS \subset X$ and $\bw \colon \bS \to \R_{\geq 0}$ given by, for 
\[ m = \poly((k+2)^{z}, 1/\eps) \]
iterations, sampling $\bx \sim \tilde{\sigma}$ i.i.d and letting $\bw(\bx) = 1/(m\tilde{\sigma}(x))$.
\item Then, with probability $1 - o(1)$ over the draw of $(\bS, \bw)$, it is an $\eps$-weak coreset for $(k,z)$-subspace approximation of $X$.
\end{itemize}
\end{theorem}

\subsection{The Important Events} 

The important events we consider mirror those the subspace approximation problem. The only event which would change is $\bE_2$, where we require $\bPi$ to be an $\eps$-subspace embedding for all subsets of $O(k^2 \log(k/\eps) / \eps) + O(\log(1/\eps) / \eps)$ points from $\bS$. This will allow us to incorporate the translation $\tau'$ from Lemma~\ref{lem:find-good-translation}.

\begin{definition}[The Events]\label{def:events-flat} Let $X = \{ x_1,\dots, x_n \} \subset \R^d$, and $\tilde{\sigma}$ the sensitivity sampling distribution of $X$ of Lemma~\ref{lem:flat-sensitivities}. We consider the following experiment,
\begin{enumerate}
\item We generate a sample $(\bS, \bw)$ by sampling from $\tilde{\sigma}$ for $m = \poly(k^z, 1/\eps)$ i.i.d iterations $\bx \sim \tilde{\sigma}$ and set $\bw(\bx) = 1/(m\tilde{\sigma}(\bx))$.
\item Furthermore, we sample $\bPi \sim \calJ_{d, t}$, which is a Johnson-Lindenstrauss map $\R^d \to \R^t$.
\item We let $\bS' = \bPi(\bS) \subset \R^t$ denote the image of $\bPi$ on $\bS$.
\end{enumerate}
The events are the following:
\begin{itemize}
\item $\bE_1$ : The weighted (multi-)set $(\bS, \bw)$ is a weak $\eps$-coreset for $(k,z)$-flat approximation of $X$ in $\R^d$.
\item $\bE_2$ : The map $\bPi \colon \R^d \to \R^t$ satisfies the following condition. For any choice of $O(k^2 \log(k/\eps) / \eps)$ points of $\bS$, $\bPi$ is an $\eps$-subspace embedding of the subspace spanned by these points.
\item $\bE_3(\beta)$ : Let $\bF' \subset \R^t$ denote the optimal $(k,z)$-flat approximation of $\bPi(X)$ in $\R^t$. Then, 
\[ \cost_z^z((\bPi(\bS), \bw), \bF') \leq \beta \cdot \cost_z^z(\bPi(X), \bF'). \]
\end{itemize}
\end{definition}

\begin{lemma}\label{lem:events-to-thm-flat}
Let $X = \{ x_1,\dots, x_n \} \subset \R^d$, and suppose $(\bS, \bw)$ and $\bPi \colon \R^d \to \R^t$ satisfy events $\bE_1$, $\bE_2$, and $\bE_3(\beta)$. Then,
\[ \min_{\substack{F' \text{ $k$-flat} \\ \text{in $\R^t$}}} \cost_z(\bPi(X), F') \geq \dfrac{1}{\beta^{1/z}(1+\eps)^4} \cdot \min_{\substack{F \text{ $k$-flat} \\ \text{in $\R^d$}}} \cost_z(X, F). \]
\end{lemma}

\begin{proof}
Consider a fixed $\bPi$ and $(\bS, \bw)$ satisfying the three events of Definition~\ref{def:events-sa}. Let $F' \subset \R^t$ be the $k$-flat which minimizes $\cost_z^z(\bPi(X), F')$. Suppose that $F'$ is specified by a $k$-dimensional subspace $R'$ and a translation $\tau'$. Then, by event $\bE_3(\beta)$, we have $\cost_z^z((\bPi(\bS), \bw), F') \leq \beta \cdot \cost_z^z(\bPi(X), F')$. Now, we apply Lemma~\ref{lem:find-good-translation} to $(\bPi(\bS), \bw)$, and we obtain a subset $Q \subset \bS$ of size $O(\log(1/\eps) / \eps)$ for which there exists a translation vector $\tau'' \in \R^t$ within the $\conv(\bPi(Q))$ such that $k$-flat $F''$ given by $\tau''$ and $R'$ satisfies
\[ \left( \sum_{x \in \bS} \bw(x) \cdot \| \bPi(x) - \rho_{F''}(\bPi(x)) \|_2^z\right)^{1/z} = \cost_z((\bPi(\bS),\bw), F'') \leq (1+\eps) \cdot \cost_z((\bPi(\bS), \bw), F'). \]
Furthermore, by Lemma~\ref{lem:find-good-subspace} to the weighted vectors $(\bPi(\bS) - \tau'', \bw)$,\footnote{Here, we are using the short-hand $\bPi(\bS) - \tau'' \eqdef \left\{ \bPi(x) - \tau'' \in \R^t : x \in \bS\right\}$.} there exists a subset $Q' \subset \bPi(\bS) - \tau''$ of size $O(k^2 \log(k/\eps) / \eps)$ and a $k$-dimensional subspace $R'' \subset \R^d$ within the span of $Q'$ such that the $k$-flat $F'''$ specified by $R''$ and $\tau''$ satisfies
\[ \left( \sum_{x \in \bS} \bw(x) \cdot \| \bPi(x) - \rho_{F'''}(\bPi(x))\|_2^z \right)^{1/z} = \cost_z((\bPi(\bS), \bw), F''') \leq (1+\eps)^2 \cdot \cost_z((\bPi(\bS), \bw), F'). \]
Recall that (i) $R''$ is a $k$-dimensional subspace lying in the span of $\bPi(Q')$, (ii) $\tau'' \in \R^t$ is within $\conv(\bPi(Q))$, and (iii) for any $x \in \bS$, $\bPi$ is an $\eps$-subspace embedding of the span of $Q \cup Q' \cup \{ x \}$. Similarly to Lemma~\ref{lem:events-to-thm-sa}, we may find a $k$-flat $U$ such that for every $x \in \bS$,
\[ \|x - \rho_{U}(x) \|_2 \leq (1+\eps) \| \bPi(x) - \rho_{F'''}(\bPi(x))\|_2, \]
and hence
\[ \cost_z((\bS, \bw), U) \leq (1+\eps) \cdot \cost_z((\bPi(\bS), \bw), F''').\]
Finally, since $(\bS, \bw)$ is a $\eps$-weak coreset, we obtain the desired inequality.
\end{proof}

As in the previous section, events $\bE_1$ and $\bE_2$ hold with sufficiently high probability. All that remains is showing that $\bE_3(1+\eps)$ holds with sufficiently high probability. We proceed in a similar fashion, where we first show a loose approximation guarantee, and later improve on it.

\begin{lemma}\label{lem:bad-approx-flat}
Fix any $\Pi \in \calJ_{d,t}$ and let $F' \subset \R^t$ denote the $k$-flat for optimal $(k,z)$-flat approximation of $\Pi(X)$ in $\R^t$. Then with probability at least $0.99$ over the draw of $(\bS, \bw)$ as per Definition~\ref{def:events-flat}, 
\[ \sum_{x \in \bS} \bw(x) \cdot \| \Pi(x) - \rho_{F'}(\Pi(x))\|_2^z \leq 100 \cdot \cost_z^z(\Pi(X), F'). \]
In other words, event $\bE_3(100)$ holds with probability at least $0.99$.
\end{lemma}

\begin{corollary}\label{cor:bad-approx-flat}
Let $X = \{ x_1,\dots, x_n \} \subset \R^d$ be any set of points. For any $\eps \in (0, 1/2)$, let $\calJ_{d,t}$ be the Johnson-Lindenstrauss map with 
\[ t \gsim \dfrac{z \cdot k^2 \cdot \polylog(k/\eps)}{\eps^3}. \]
Then, with probability at least $0.97$ over the draw of $\bPi \sim \calJ_{d,t}$,
\[\dfrac{1}{100^{1/z}(1+\eps)^4} \cdot \min_{\substack{F \subset \R^d \\ \text{$k$-flat}}} \cost_z(X, F) \leq \min_{\substack{F' \subset \R^t \\ \text{$k$-flat}}} \cost_z(\bPi(X), F'). \]
\end{corollary}

\subsubsection{Improving the approximation}

The improvement of the approximation, follows from upper bounding the variance, as in the $(k,z)$-clustering problem, and the $(k,z)$-subspace approximation problem. In particular, we show that $\bE_3(1+\eps)$ holds. Fix $X = \{ x_1,\dots, x_n\} \subset \R^d$ and $F \subset \R^d$ be the optimal $(k,z)$-flat approximation of $X$ in $\R^d$. The sensitivity function $\sigma \colon X \to \R_{\geq 0}$ specified in Lemma~\ref{lem:flat-sensitivities} specify the sensitivity sampling distribution $\tilde{\sigma}$. 

We let $\bE_4$ denote the following event with respect to the randomness in $\bPi \sim \calJ_{d,t}$. For each $x \in X$, we let $\bD_x \in \R_{\geq 0}$ denote the random variable
\[ \bD_x \eqdef \dfrac{\|\bPi(x) - \bPi(\rho_F(x))\|_2}{\|x - \rho_F(x)\|_2},\]
and as in (\ref{eq:def-D-sa}) and (\ref{eq:event-4-sa}), event $\bE_4$, which occurs with probability at least $0.99$, whenever
\[ \sum_{x \in X} \bD_x^{2z} \cdot \sigma(x) \leq 100 \cdot 2^z \cdot \frakS_{\sigma}. \]

\begin{lemma}\label{lem:variance-bound-flat}
Let $\Pi \in \calJ_{d,t}$ be a Johnson-Lindensrauss map where, for $\alpha > 1$, the following events hold:
\begin{enumerate}
\item Guarantee from Lemma~\ref{lem:flat-ub}: $\sum_{x \in X} \| \Pi(x) - \Pi(\rho_F(x))\|_2^z \leq \alpha \cdot \cost_z^z(X,F)$.
\item Guarantee from Corollary~\ref{cor:bad-approx-flat}: letting $F' \subset \R^t$ be the optimal $(k,z)$-flat approximation of $\Pi(X)$, then $\cost_z^z(X,F) \leq \alpha \cost_z^z(\Pi(X), F')$.
\item Event $\bE_4$ holds.
\end{enumerate}
Then, if we let $(\bS, \bw)$ denote $m = \poly(k^z, 1/\eps, \alpha)$ i.i.d draws from $\tilde{\sigma}$ and $\bw(x) = 1 / (m \tilde{\sigma}(x))$, with probability at least $0.99$, 
\[ \cost_z^z((\Pi(\bS), \bw), F') \leq (1+\eps) \cdot \cost_z^z(\Pi(X), F'). \]
\end{lemma}

\begin{proof}
We similarly bound the variance of 
\begin{align*}
\Var_{\bS}\left[ \Ex_{\bx \sim \bS}\left[ \dfrac{1}{\tilde{\sigma}(\bx)} \dfrac{\| \Pi(\bx) - \rho_{F'}\left(\Pi(\bx)\right) \|_2^z}{\cost_z^z(\Pi(X), F')} \right] \right] \leq \dfrac{\frakS_{\sigma}}{m} \sum_{x \in X} \left(\dfrac{1}{\sigma(x)} \cdot \dfrac{\|\Pi(x) - \rho_{F'}(\Pi(x))\|_2^z}{\cost_z^{2z}(\Pi(X), F')} \right).
\end{align*}
It is not hard to show, as in the proof of Lemma~\ref{lem:variance-bound-sa}, that writing $y_x = \rho_F(x) \in \R^d$ and $Y = \{ y_x : x \in X\}$, that
\begin{align*}
\dfrac{\|\Pi(x) - \rho_{F'}(\Pi(x))\|_2^z}{\cost_z^z(\Pi(X), F')} \leq 2^{z-1}  \alpha  \bD_x^{z} \cdot \dfrac{\|x - \rho_F(x)\|_2^z}{\cost_z^z(X, F)} + 2^{2z-2} (1+\alpha^2) \sup_{\substack{v \in \R^t \\ \mu \in \R}} \dfrac{|\langle \Pi(y_x), v \rangle - \mu|^z}{\sum_{x' \in X} |\langle \Pi(y_{x'}), v \rangle - \mu|^z},
\end{align*}
and similarly to before, we have
\begin{align*}
\sup_{\substack{v \in \R^t \\ \mu \in \R}} \dfrac{|\langle \Pi(y_x), v \rangle - \mu|^z}{\sum_{x'\in X} |\langle \Pi(y_{x'}), v \rangle - \mu|^z} \leq \sup_{\substack{u \in \R^d \\ \phi \in R}} \dfrac{|\langle y_x, u \rangle - \phi|^z}{\sum_{x'\in X} |\langle y_{x'}, u \rangle - \phi|^z}.
\end{align*}
This implies that the variance is at most
\[ 2^{4z-2} \cdot \dfrac{\frakS_{\sigma} \alpha^4}{m} \cdot \sum_{x \in X} \bD_x^{2z} \cdot \sigma(x) \leq \eps^2\]
by setting $m = \poly((k+2)^z, 1/\eps, \alpha)$ to be large enough when $\bE_4$ holds, and we apply Chebyshev's inequality.
\end{proof}

\begin{corollary}
Let $X = \{ x_1,\dots, x_n \} \subset \R^d$ be any set of points, and let $F \subset \R^d$ be the optimal $(k,z)$-flat approximation of $X$. For any $\eps \in (0,1/2)$, let $\calJ_{d,t}$ be the Johnson-Lindenstrauss map with 
\[ t \gsim \dfrac{z \cdot k^2 \cdot \polylog(k/\eps)}{\eps^3}.\]
Then, with probability at least $0.92$ over the draw of $\bPi \sim \calJ_{d,t}$,
\[ \dfrac{1}{(1+\eps)^{4+1/z}} \cdot \cost_z(X, F) \leq \min_{\substack{F' \subset \R^t \\ \text{$k$-flat}}} \cost_z(\bPi(X), F'). \] 
\end{corollary}

\section{$k$-Line Approximation}\label{sec:line}

In the $(k,z)$-\emph{line approximation} problem, we consider a collection of $k$ lines in $\R^d$. A line is encoded by a vector $v \in \R^d$ and a unit vector $u \in S^{d-1}$, where we will write
\[ \ell(v, u) = \left\{ v + t \cdot u \in \R^d : t \in \R\right\}. \]
For a single line $\ell$ encoded by $v$ and $u$, we write $\rho_{\ell} \colon \R^d \to \R$ as the orthogonal projection of a point onto $\ell$, i.e., the closest vector which lies on the line, where
\[ \rho_{\ell}(x) = \argmin_{y \in \ell} \| x - y\|_2^2 = v + \langle x - v, u \rangle u.  \]
For any set of $k$ lines, $L = \{ \ell_1, \dots, \ell_k \}$, and a point $x \in X$, we write
\[ \cost_z^z(x, L) = \min_{\ell \in L} \|x - \rho_{\ell}(x)\|_2^z,\]
and for any dataset $X \subset \R^d$ and set of lines $L$, we consider the map $\ell \colon X \to L$ which sends a point $x$ to its nearest line in $L$. Then, we write
\[ \cost_z^z(X, L) = \sum_{x \in X} \| x - \rho_{\ell(x)}(x) \|_2^z, \]
as the cost of representing the points in $X$ according to the $k$ lines in $L$. In this section, we show that we may find the optimal $(k,z)$-line approximation after applying a Johnson-Lindenstrauss map. Specifically, we prove:
\begin{theorem}[Johnson-Lindenstrauss for $(k,z)$-Line Approximation]\label{thm:k-line} Let $X = \{ x_1,\dots, x_n \} \subset \R^d$ be any set of points, and let $L = \{ \ell_1,\dots, \ell_k \}$ denote a set of lines in $\R^d$ for optimally $(k,z)$-line approximation of $X$. For any $\eps \in (0, 1)$, suppose we let $\calJ_{d,t}$ be a distribution over Johnson-Lindenstrauss maps where
\[ t \gsim \dfrac{k \log \log n + z + \log(1/\eps)}{\eps^3}. \]
Then, with probability at least $0.9$ over the draw of $\bPi \sim \calJ_{d,t}$,
\[ \dfrac{1}{1+\eps} \cdot \cost_z(X, L) \leq \min_{\substack{L' \text{ $k$ lines} \\ \text{in $\R^t$}}} \cost_z(\bPi(X), L') \leq (1+\eps) \cdot \cost_z(X, L).\]
\end{theorem}

\subsection{Easy Direction: Optimum Cost Does Not Increase}

\begin{lemma}\label{lem:upper-bound-lines}
Let $X = \{ x_1,\dots, x_n \} \subset \R^d$ be any set of points and let $L = \{ \ell_1,\dots, \ell_k \}$ be a set of $k$ lines in $\R^d$ for optimal $(k,z)$-line approximation of $X$, and for each $x \in X$, let $\ell(x) \in L$ be the line assigned to $x$. We let $\calJ_{d,t}$ be the distribution over Johnson-Lindenstrauss maps. If $t \gsim z / \eps^2$, then with probability at least $0.99$ over the draw of $\bPi \sim \calJ_{d,t}$,
\[ \left(\sum_{x \in X} \| \bPi(x) - \bPi(\rho_{\ell(x)}(x)) \|_2^z\right)^{1/z} \leq (1+\eps) \cdot \cost_z(X, L), \]
and hence,
\[ \min_{\substack{L' \text{ $k$ lines} \\ \text{in $\R^t$}}} \cost_z(\bPi(X), L') \leq (1+\eps) \cdot \cost_z(X, L).\]
\end{lemma}

By now, there is a straight-forward way to prove the above lemma. For each set of $k$ lines $L = \{ \ell_1,\dots, \ell_k\}$ in $\R^d$, there is an analogous definition of $k$ lines $\bPi(L)$ in $\R^t$. Hence, we use Lemma~\ref{lem:gaussians} as in previous sections to upper bound the cost $\cost_z(\bPi(X), \bPi(L))$.


\subsection{Hard Direction: Optimum Cost Does Not Decrease}

\subsubsection{Preliminaries}

At a high level, we proceed with the same argument as in previous sections: we consider a sensitivity function for $(k,z)$-line approximation of $X$ in $\R^d$, and use it to build a weak coreset, as well as argue that sensitivity sampling is a low-variance estimator of the optimal $(k,z)$-line approximation in the projected space. The proof in this section will be significantly more complicated than the previous section. Defining the appropriate sensitivity functions, which will give a low-variance estimator in the projected space, is considerably more difficult than the expressions of Lemmas~\ref{lem:sens},~\ref{lem:subspace-sensitivities}, and~\ref{lem:flat-sensitivities}. For this reason, will be proceed by assuming access to a sensitivity function which we will define lated in the section.

\begin{definition}[Weak Coresets for $(k,z)$-Line Approximation] Let $X = \{ x_1,\dots, x_n \} \subset \R^d$ be a set of points. A weak $\eps$-coreset of $X$ for $(k,z)$-line approximation is a weighted subset $S \subset \R^d$ of points with weights $w \colon S \to \R_{\geq 0}$ such that
\begin{align*}
\frac{1}{1+\eps} \cdot \min_{\substack{L \text{ $k$ lines} \\ \text{in $\R^d$}}} \cost_z(X, L) \leq \min_{\substack{L \text{ $k$ lines} \\ \text{in $\R^d$}}} \cost_z((S, w), L) \leq (1+\eps) \cdot \min_{\substack{L \text{ $k$ lines} \\ \text{in $\R^d$}}} \cost_z(X, L) 
\end{align*}
\end{definition}

\begin{definition}[Sensitivities]
Let $n, d \in \N$, and consider any set of points $X = \{ x_1,\dots, x_n \} \subset \R^d$, as well as $k \in \N$ and $z\geq 1$. A sensitivity function $\sigma \colon X \to \R_{\geq 0}$ for $(k,z)$-line approximation in $\R^d$ is a function which satisfies that, for all $x \in X$,
\begin{align*}
\sup_{\substack{L \text{: $k$ lines} \\ \text{in $\R^d$}}} \dfrac{\| x - \rho_{\ell(x)}(x)\|_2^z}{\cost_z^z(X, L)} \leq \sigma(x).
\end{align*}
The total sensitivity of the sensitivity function $\sigma$ is given by
\[ \frakS_{\sigma} = \sum_{x \in X} \sigma(x). \]
For a sensitivity function, we let $\tilde{\sigma}$ denote the sensitivity sampling distribution, supported on $X$, which samples $x \in X$ with probability proportional to $\sigma(x)$.
\end{definition}

Similarly to before, we first give a lemma which narrows down the space of the optimal line approximations for a set of points. the following lemma is a re-formulations of Lemma~\ref{lem:find-good-subspace} and Lemma~\ref{lem:find-good-translation} catered to the case of $(k,z)$-line approximation.

\begin{lemma}[Theorem 3.1 and Lemma~3.3 of~\cite{SV12}]\label{lem:within-span}
Let $d \in \N$ and $S \subset \R^d$ be any set of points with weights $w \colon S \to \R_{\geq 0}$, $\eps \in (0, 1/2)$, and $z \geq 1$. There exists a subset $Q \subset S$ of size $O(\log(1/\eps) / \eps)$ and a line $\ell$ in $\R^d$ within the span of $Q$ such that
\[ \cost_{z}((S, w), \{\ell\}) \leq (1+\eps) \min_{\substack{\ell' \text{ line} \\ \text{in $\R^d$}}} \cost_z((S, w), \{\ell'\}). \]
\end{lemma}

\begin{lemma}[Weak Coresets for $k$-Line Approximation~\cite{FL11, VX12}]
For any subset $X = \{ x_1,\dots, x_n \} \subset \R^d$ and $\eps \in (0, 1/2)$, let $\sigma$ denote a sensitivity function for $(k,z)$-line approximation of $X$ with total sensitivity $\frakS_{\sigma}$ and let $\tilde{\sigma}$ its sensitivity sampling distribution. 
\begin{itemize}
\item Let $(\bS, \bw)$ denote the random (multi-)set $\bS \subset X$ and $\bw \colon \bS \to \R_{\geq 0}$ given by, for 
\[ m = \poly(\frakS_{\sigma}, k, 1/\eps),\]
iterations, sampling $\bx \sim \tilde{\sigma}$ i.i.d and letting $\bw(\bx) = 1/(m \tilde{\sigma}(\bx))$.
\item Then, with probability $1 - o(1)$ over the draw of $(\bS, \bw)$, it is an $\eps$-weak coreset for $(k,z)$-line approximation of $X$.
\end{itemize}
\end{lemma}

We note that \cite{FL11} and \cite{VX12} only give a strong coreset for $(k,z)$-line approximation of $\poly(\frakS_{\sigma}, k,d,1/\eps)$. For example, Theorem~13 in \cite{VX12} giving the above bound follows from the fact that the ``function dimension'' (see Definition~3 of \cite{VX12}) for $(k,z)$-line approximation is $O(kd)$. However, Lemma~\ref{lem:within-span} implies that for any set of points, a line which approximates the points is within a span of $O(\log(1/\eps)/\eps)$ points. This means that, for $\eps$-weak coresets, it suffices to only consider $k$ lines spanned by $O(k\log(1/\eps)/\eps)$, giving us a ``function dimension'' of $O(k \log(1/\eps)/\eps)$.

\subsubsection{The Important Events}

\begin{definition}[The Events]\label{def:k-line-events}
Let $X = \{ x_1,\dots, x_n \} \subset \R^d$, and $\sigma$ be a sensitivity function for $(k,z)$-line approximation of $X$ in $\R^d$, with total sensitivity $\frakS_{\sigma}$ and sensitivity sampling distribution $\tilde{\sigma}$. We consider the following experiment,
\begin{enumerate}
\item We generate a sample $(\bS, \bw)$ by sampling from $\tilde{\sigma}$ for $m = \poly(\frakS_{\sigma}, k, 1/\eps)$ i.i.d iterations $\bx \sim \tilde{\sigma}$ and set $\bw(\bx) = 1/(m\tilde{\sigma}(\bx))$. 
\item Furthermore, we sample $\bPi \sim \calJ_{d,t}$, which is a Johnson-Lindenstrauss map $\R^d \to \R^t$.
\item We let $\bS' = \bPi(\bS) \subset \R^t$ denote the image of $\bPi$ on $\bS$.
\end{enumerate}
The events are the following:
\begin{itemize}
\item $\bE_1$ : The weighted (multi-)set $(\bS, \bw)$ is a weak $\eps$-coreset for $(k,z)$-line approximation of $X$ in $\R^d$. 
\item $\bE_2$ : For any subset of $O(\log(1/\eps)/\eps)$ points from $\bS$, the map $\bPi \colon \R^d \to \R^t$ is an $\eps$-subspace embedding for the subspace spanned by that subset.
\item $\bE_3(\beta)$ : Let $\bL' = \{ \bell_1',\dots,\bell_k'\}$ denote $k$ lines in $\R^t$ for optimal $(k,z)$-line approximation of $\bPi(X)$ in $\R^t$. Then,
\begin{align*}
\cost_z^z( (\bPi(\bS), \bw), \bL') \leq \beta\cdot \cost_z^z(\bPi(X), \bL').
\end{align*}
\end{itemize}
\end{definition}

\begin{lemma}\label{lem:lines-events-to-approx}
Let $X = \{ x_1,\dots, x_n \} \subset \R^d$, and suppose $(\bS, \bw)$ and $\bPi \colon \R^d \to \R^t$ satisfy events $\bE_1, \bE_2$ and $\bE_3$. Then,
\begin{align*}
\min_{\substack{L' \text{ $k$ lines} \\ \text{in $\R^t$}}} \cost_z(\bPi(X), L') \geq \dfrac{1}{\beta^{1/z}(1+\eps)^{3}} \cdot \min_{\substack{L \text{ $k$ lines} \\ \text{in $\R^d$}}} \cost_z(X, L).
\end{align*}
\end{lemma}

\begin{proof}
Let $\bPi \sim \calJ_{d,t}$ and $(\bS, \bw)$ be sampled according to Definition~\ref{def:k-line-events}, and suppose events $\bE_1, \bE_2$ and $\bE_3$ all hold. Let $\bL' = \{ \bell_1',\dots, \bell_k'\}$ denote the set of $k$ lines for optimal $(k,z)$-line approximation of $\bPi(X)$ in $\R^t$. Then, by event $\bE_3$, we have $\cost_z^z((\bPi(\bS), \bw), \bL') \leq \beta \cdot \cost_z^z(\bPi(X), \bL')$. Consider the partition of $\bS$ into $\bS_1,\dots, \bS_k$ induced by the lines in $\bL'$ closest to $\bPi(\bS)$. 

For each $i \in [k]$, we apply Lemma~\ref{lem:within-span} to $\bPi(\bS_i)$ with weights $\bw \colon \bS_i \to \R_{\geq 0}$. In particular, there exists subsets $\bQ_1 \subset \bS_1,\dots , \bQ_k \subset \bS_k$ and $k$ lines $\bL'' = \{ \bell_1'', \dots, \bell_k''\}$ in $\R^t$ such that each line $\bell_i''$ lie in the span of $\bQ_i$, and 
\[ \cost_z((\bPi(\bS), \bw), \bL'') \leq (1+\eps) \cdot \cost_z((\bPi(\bS), \bw), \bL').\]
Event $\bE_2$ implies that for each $i \in [k]$ and each $x \in \bS$, $\bPi$ is an $\eps$-subspace embedding for the subspace spanned by $\bQ_i \cup \{ x \}$. It is not hard to see, that there exists $k$ lines $\bH = \{ \bh_1,\dots, \bh_k\}$ in $\R^d$ such that for all $x \in \bS$, 
\[ \| x - \rho_{\bh_i}(x) \|_2 \leq (1+\eps) \cdot \| \bPi(x) - \rho_{\bell_i''}(\bPi(x))\|_2, \]
and therefore, 
\begin{align*}
\cost_z((\bS, \bw), \bH) \leq (1+\eps) \cdot \cost_z((\bPi(\bS), \bw), \bL'').
\end{align*}
Lastly, $(\bS, \bw)$ is a $\eps$-weak coreset for $X$, which means that 
\[ \min_{\substack{L \text{ $k$ lines} \\ \text{in $\R^d$}}} \cost_z(X, L) \leq (1+\eps) \cdot \cost_z((\bS, \bw), \bH).\]
Combining all inequalities gives the desired lemma.
\end{proof}

By now, we note that it is straight-forward to prove the following corollary, which gives a dimension reduction bound which depends on the total sensitivity of a sensitivity function.
\begin{corollary}\label{cor:bad-approx-lines}
Let $X = \{ x_1,\dots, x_n \} \subset \R^d$ be any set of points, and for $k \in \N$ and $z \geq 1$, let $\sigma \colon X \to \R_{\geq 0}$ be a sensitivity function for $(k,z)$-line approximation of $X$ in $\R^d$. For any $\eps \in (0, 1/2)$, let $\calJ_{d,t}$ be the Johnson-Lindenstrauss map with
\[ t \gsim \dfrac{\log(\frakS_{\sigma}, k, 1/\eps)}{\eps^3}. \]
Then, with probability at least $0.97$ over the draw of $\bPi \sim \calJ_{d,t}$, 
\begin{align*}
\dfrac{1}{100^{1/z} (1+\eps)^3} \min_{\substack{L \text{ $k$ lines}\\ \text{in $\R^d$}}} \cost_z(X, L) \leq \min_{\substack{L' \text{ $k$ lines}\\ \text{in $\R^t$}}} \cost_z(\bPi(X), L') 
\end{align*}
\end{corollary}


\subsection{A Sensitivity Function for $(k,z)$-Line Approximation}

We now describe a sensitivity function for $(k,z)$-line approximation of points in $\R^d$. Similarly to the previous section, we consider a set of points $X = \{x_1,\dots, x_n \} \subset \R^d$, and we design a sensitivity function $\sigma \colon X \to \R_{\geq0}$ for $(k,z)$-line approximation of $X$ in $\R^d$. The sensitivity function should satisfy two requirements. The first is that we have a good bound on the total sensitivity, $\frakS_{\sigma}$, where the target dimension $t$ will have logarithmic dependence on $\frakS_{\sigma}$ (for example, like in Corollary~\ref{cor:bad-approx-lines}). 

The second is that $\bE_3(1+\eps)$ will hold with sufficiently high probability over the draw of $\bPi \sim \calJ_{d,t}$. In other words, we will proceed similarly to Lemmas~\ref{lem:variance-bound},~\ref{lem:variance-bound-sa}, and~\ref{lem:variance-bound-flat} and show that, for the optimal $(k,z)$-line approximation $\bL'$ of $\bPi(X)$ in $\R^t$, sampling according to the sensitivity sampling distribution gives a low-variance estimate for the cost of $\bL'$. 

\subsubsection{From Coresets for $(k, \infty)$-line approximation to Sensitivity Functions}

Unfortunately, we do not know of a ``clean'' description of a sensitivity function for $(k,z)$-line approximation, as was the case in previous definitions. Certainly, one may define a sensitivity function to be $\sigma(x) = \sup_{L} \cost_z^z(x, L) / \cost_z^z(X, L)$, but then arguing that $\bE_3(1+\eps)$ holds with high probability becomes more complicated. The sensitivity function which we present follows the connection between sensitivity and $\ell_{\infty}$-coresets \cite{VX12b}.

\begin{definition}[$c$-coresets for $(k,\infty)$-line approximation]\label{def:inf-coreset}
Let $Y = \{ y_1,\dots, y_n \} \subset \R^d$ be any subset of points, and $c \geq 1$. A subset $A \subset Y$ is a $c$-coreset  for $(k, \infty)$-line approximation if the following holds:
\begin{itemize}
\item Let $L = \{ \ell_1,\dots, \ell_{k} \}$ be any collection of $k$ lines in $\R^d$, and $r \in \R_{\geq 0}$ such that for all $y \in A$,
\[ \min_{\ell \in L} \| y - \rho_{\ell}(y)\|_2 \leq r. \]
\item Then, for all $x \in X$, 
\[ \min_{\ell \in L} \| y - \rho_{\ell}(y) \|_2 \leq c r. \]
\end{itemize}
\end{definition}

Note that the $(k, \infty)$-line approximation is the problem of minimum enclosing cylinder: we are given a set of points $Y$, and want to find a set of $k$ cylinders $C_1,\dots, C_k \subset \R^d$ of smallest radius such that $Y \subset \bigcup_{i=1}^k C_i$. Thus, Definition~\ref{def:inf-coreset}, a set $A \subset Y$ is a $c$-coreset for $(k, \infty)$-line approximation if, given any $k$ cylinders which contain $A$, increasing the radii by a factor of $c$ contains $Y$. The reason they will be relevant for defining a sensitivity function is the following simple lemma, whose main idea is from \cite{VX12b}.

\begin{lemma}[Sensitivities from $c$-coresets for $(k,\infty)$-line approximation (see Lemma~3.1 in~\cite{VX12b})]\label{lem:lines-sensitivity}
Let $X = \{ x_1,\dots, x_n \} \subset \R^d$ be any set of points and $k \in  \N$, $z \geq 1$. Let $L = \{\ell_1,\dots, \ell_k \}$ be the $k$ lines in $\R^d$ for optimal $(k,z)$-line approximation of $X$, and let $Y = \{ y_x \in \R^d : x \in X\}$ where $y_x= \rho_{\ell(x)}(x)$. For $c \geq 1$, let the function $\sigma \colon X \to \R_{\geq 0}$ be defined as follows:
\begin{itemize}
\item Let $A_1, A_2,\dots, A_{s}$ denote a partition of $Y$ where each $A_i$ is a $c$-coreset for $(k,\infty)$-line approximation of $Y \setminus \left(\bigcup_{i' = 1}^{i-1} A_{i'}\right)$. 
\item For each $x \in X$, where $y_x \in A_i$ we let
\[ \sigma(x) \eqdef 2^{z-1} \cdot \dfrac{\|x - y_x \|_2^z}{\cost_z^z(X, L)} + 2^{2z-1} \cdot \dfrac{c}{i}. \]
\end{itemize}
Then, $\sigma$ is a sensitivity function for $(k,z)$-line approximation, and the total sensitivity 
\[ \frakS_{\sigma} = O\left(2^{2z} \cdot c \cdot \log n \cdot \max_{i \in [s]} |A_{i}|\right)  \]
\end{lemma} 

\begin{proof}
Suppose $x \in X$ and $y_x\in A_i$. Consider any set $L' = \{ \ell_1',\dots, \ell_k' \}$ of $k$ lines in $\R^d$. The goal is to show that
\begin{align}
\dfrac{\min_{\ell' \in L'}\|x - \rho_{\ell'}(x)\|_2^z}{\cost_z^z(X, L')} \leq \sigma(x). \label{eq:sensitivity-requirements}
\end{align}
We will first use H\"{o}lder inequality and the triangle inequality, as well as the fact that $y_x \in A_i$ in order to write the following:
\begin{align}
\min_{\ell' \in L'}\|x - \rho_{\ell'}(x)\|_2^z &\leq 2^{z-1} \| x - y_x\|_2^z + 2^{z-1} \min_{\ell' \in L'} \| y_x - \rho_{\ell'}(y_x)\|_2^z \nonumber \\
		&\leq 2^{z-1} \|x -y_x\|_2^z + 2^{z-1} \left( \cost_z^z(Y, L') \cdot \dfrac{c}{i}\right). \label{eq:ell-infty-coreset}
\end{align}
The justification for (\ref{eq:ell-infty-coreset}) is the following: for every $j \leq i$, $A_j$ is a $c$-coreset for $(k,\infty)$-line approximation of a set of points which contains $y_x$. Therefore, if $\min_{\ell ' \in L'} \| y_x - \rho_{\ell'}(y_x)\|_2 = r$, there must exists a point $u \in A_j$ with $\min_{\ell' \in L'} \| u - \rho_{\ell'}(u) \|_2 \geq r / c$. Suppose otherwise: every $u \in A_j$ satisfies $\min_{\ell' \in L'} \| u - \rho_{\ell'}(u)\|_2 < r/c$. Then, the $k$ cylinders of radius $r/c$ contain $A_j$, so increasing the radius by a factor of $c$ contains $y_x$. However, this means $\min_{\ell' \in L'} \|y_x -\rho_{\ell'}(y_x) \|_2 < c \cdot r/c < r$, which is a contradiction.

Hence, we always have that $y_x \in A_i$ satisfies
\[ \min_{\ell' \in L'} \| y_x - \rho_{\ell'}(y_x)\|_2^z \leq \cost_z^z(Y, L') \cdot \dfrac{c}{i}. \] 
Continuing on upper-bounding (\ref{eq:ell-infty-coreset}), we now use the fact $\cost_z^z(Y, L') \leq 2^{z-1} \cdot \cost_z^z(X, L) +  2^{z-1} \cdot \cost_z^z(X, L') \leq 2^{z} \cost_z^z(X, L')$ because $\cost_z^z(X, L)$ is the optimal $(k,z)$-line approximation. Therefore,
\begin{align*}
\min_{\ell' \in L'}\| x - \rho_{\ell'}(x)\|_2^z \leq 2^{z-1} \cdot \|x - y_x \|_2^z + 2^{2z-1} \cdot \dfrac{c}{i} \cdot \cost_z^z(X, L'),
\end{align*}
so dividing by $\cost_z^z(X, L')$ and noticing that $\cost_z^z(X, L) \leq \cost_z^z(X, L')$ implies $\sigma$ is a sensitivity function.

The bound on total sensitivity then follows from
\begin{align*}
\frakS_{\sigma} = \sum_{x \in X} \sigma(x) = \sum_{i =1}^s \sum_{x \in A_{s'}} \sigma(x) = 2^{z-1} + 2^{2z-1} \sum_{i = 1}^s |A_{i}| \cdot \dfrac{c}{s'} = O\left( 2^{2z}\cdot c \cdot \max_{i \in [s]} |A_i| \cdot\log n\right),
\end{align*}
since $s \leq n$.
\end{proof}

\subsubsection{A simple coreset for $(k,\infty)$-line approximation of one-dimensional instances}

Suppose first, that a dataset $Y = \{ y_1,\dots, y_n \}$ lies on a line in $\R^d$, and let $C_1,\dots, C_k$ be a collection of $k$ cylinders. Then, the intersection of the cylinders with the line results in a union of $k$ intervals on the line. If we increase the radius of each cylinder $C_1,\dots, C_k$ by a factor of $c$, the lengths of the intervals are scaled by factor of $c$ (while keeping center of interval fixed). We first show that, for any $Y = \{ y_1,\dots, y_n \}$ which lie on a line, there exists a small subset $Q \subset Y$ such that: if $I_1,\dots, I_k$ is any collection of $k$ intervals which covers $Q$, then increasing the length of each interval by a factor of $3$ (while keeping the center of the interval fixed) covers $Y$. 

\begin{lemma}\label{lem:line-coreset}
There exists a large enough constant $c_1 \in \R_{\geq 0}$ such that the following is true. Let $Y = \{ y_1,\dots, y_n \}$ be a set of points lying on a line in $\R^d$, and $k \in \N$. There exists a subset $Q \subset Y$ which is a $3$-coreset for $(k,\infty)$-line approximation of size at most $(c_1 \log n)^{k}$.
\end{lemma}

\begin{proof}
The construction is recursive. Let $\ell$ be the line containing $Y$, and after choosing an arbitrary direction on $\ell$, let $y_1, \dots, y_n$ be the points in sorted order according to the chosen direction. 

The set $Q$ is initially empty, and we include $Q \leftarrow \{ y_1, y_{\lceil n/2 \rceil} ,  y_n \}$. Suppose that $\| y_1 - y_{\lceil n / 2\rceil} \|_2 \geq \| y_{\lceil n / 2 \rceil} - y_n \|_2$ (the construction is symmetric, with $y_1$ and $y_{n}$ switched otherwise). We divide $Y$ into two sets, the subsets $Y_{L} = \{ y_1,\dots, y_{\lfloor n/2 \rfloor} \}$ and $Y_{R} = \{ y_{\lceil n / 2 \rceil}, \dots, y_n \}$. Then, we perform three recursive calls: (i) we let $Q_1$ be a $3$-coreset for $(k, \infty)$-line approximation of $Y_{L}$, (ii) we let $Q_2$ be a $3$-coreset for $(k-1, \infty)$-line approximation of $Y_L$, and (iii) we let $Q_3$ be a $3$-coreset for $(k-1, \infty)$-line approximation of $Y_R$. We add $Q_1, Q_2$, and $Q_3$ to $Q$.

The proof of correctness argues as follows. Let $C_1,\dots, C_k$ be an arbitrary collection of $k$ cylinders which covers $Q$. The goal is to show that increasing the radius of $C_1,\dots, C_k$ by a factor of 3 covers $Y$. Let $I_1,\dots, I_k$ be the intervals given by $I_i = \ell \cap C_i$. We let the indices $u,v \in [k]$ be such that $I_u$ is the first interval which contains $y_1$, and $I_v$ the last interval which contains $y_n$. We note that $I_1 \cup \dots \cup I_k$ covers $Q$. We must show that if we increase the length of each interval by a factor of $3$, we cover $Y$. We consider three cases:
\begin{itemize}
\item Suppose there exists an index $i^* \in [k]$ such that $y_1$ and $y_{\lceil n / 2 \rceil}$ both lie in the interval $I_{i^*}$. Recall $\| y_1 - y_{\lceil n / 2 \rceil}\|_2 \geq \| y_{\lceil n / 2 \rceil} - y_n \|_2$, and all points are contained within $y_1$ and $y_n$. Hence, when we increase the length of $I_{i^*}$ by a factor of $3$ while keeping center fixed, $y_1$ and $y_n$ lie in the same interval, and thus cover $Y$.
\item Suppose $y_1$ and $y_{\lceil n / 2 \rceil}$ lie in different intervals, but there exists $i^*$ such that $y_{\lceil n / 2\rceil}$ and $y_{n}$ lie in the interval $I_{i^*}$. Then, since all points of $Y_{R}$ between $y_{\lceil n / 2\rceil}$ and $y_n$, $I_{i^*}$ covers $Y_{R}$. Since $I_1,\dots, I_k$ covers $Q_1$ and $Q_1$ is a $3$-coreset for $(k, \infty)$-line approximation of $Y_L$, increasing the length of each interval by a factor of $3$ covers $Y_L$, and therefore all of $Y$.
\item Suppose $y_1$, $y_{\lceil n / 2\rceil}$ and $y_n$ all lie in different intervals. Then, since $y_1$ and $y_{\lceil n / 2 \rceil}$ are not on the same interval, the $k-1$ intervals $\bigcup_{i \in [k]\setminus \{u\}} I_i$ covers $Q_2$. Similarly, $y_{\lceil n /2\rceil}$ and $y_{n}$ are not on the same interval, so the $k-1$ intervals $\bigcup_{i \in [k]\setminus\{v\}} I_i$ covers $Q_3$. Since $Q_2$ is a $3$-coreset for $(k-1, \infty)$-line approximation of $Y_L$, increasing the radius of each interval by a factor of $3$ covers all of $Y_L$. In addition, $Q_3$ is a $3$-coreset for $(k-1, \infty)$-line approximation of $Y_R$, so increasing length of intervals by a factor of $3$ covers $Y_R$.
\end{itemize}
This concludes the correctness of the coreset, and it remains to upper bound the size. Let $f(k, n) \in \N$ be an upper bound on the coreset size of $(k, \infty)$-line approximation of a subset of size $n$. We have $f(1, n) = 2$, since any single interval which covers $y_1$ and $y_n$ covers everything in between them. By our recursive construction, we have
\[ f(k, n) \leq 3 + f(k, n/2) + 2 \cdot f(k-1, n/2).\]
By a simple induction, one can show $f(k, n)$ is at most $(c_1 \log n)^{k}$ when $k \geq 2$, for large enough constant $c_1$ and large enough $n$.
\end{proof}

\subsubsection{The coreset for points on $k$ lines and the effect of dimension reduction}

\begin{lemma}\label{lem:coreset-for-lines}
There exists a large enough constant $c_1 \in \R_{\geq 0}$ such that the following is true. Let $Y = \{ y_1,\dots, y_n \}$ be a set of points lying on $k$ lines in $\R^d$. There exists a subset $Q \subset Y$ which satisfies the following two requirements:
\begin{enumerate}
\item\label{en:coreset-in-orig} $Q$ is a $3$-coreset for $(k, \infty)$-line approximation of $Y$ size at most $k (c_1 \log n)^{k}$.
\item\label{en:coreset-after-map} If $\Pi \colon \R^d \to \R^t$ is a linear map, then $\Pi(Q)$ is a $3$-coreset for $(k,\infty)$-line approximation of $\Pi(Y)$. 
\end{enumerate}
\end{lemma}

\begin{proof}
Let $Y_1,\dots, Y_k$ be the partition of $Y$ into points lying on the lines $\ell_1,\dots, \ell_k$ of $\R^d$, respectively. We may write each line $\ell_i$ by two vectors $u_i, v_i \in \R^d$, and have
\[ \ell_i = \left\{ u_i + t \cdot v_i : t \in \R \right\}.\]
Let $Q_i$ be the $3$-coreset for $(k,\infty)$-line approximation of $Y_i$ specified by Lemma~\ref{lem:line-coreset}. We let $Q$ be the union of all $Q_i$. Item~\ref{en:coreset-in-orig} follows from Lemma~\ref{lem:line-coreset}, since we are taking the union of $k$ coresets. 

We now argue Item~\ref{en:coreset-after-map}. Since $\Pi$ is a linear map, and every point in $Y_i$ lies on the line $\ell_i$, there exists a map $t \colon Y_i \to \R$ where each $y \in Y_i$ satisfies
\[ y = u_i + t(y) \cdot v_i \in \R^d\qquad\text{and thus,}\qquad \Pi(y) = \Pi(u_i) + t(y) \cdot \Pi(v_i) \in \R^t. \]
In other words, $\Pi(Y_i)$ lies within a line in $\R^t$. We note that the relative order of points in $\Pi(Y_i)$ remains the same, since for any two points $y, y' \in Y_i$,
\[ \| \Pi(y) - \Pi(y')\|_2 = | t(y) - t(y')| \cdot \| \Pi(v_i) \|_2, \qquad \| y - y' \|_2 = |t(y) - t(y')| \cdot \| v_i \|_2. \] 
We note that the construction of Lemma~\ref{lem:line-coreset} only considers the order of points in $Y_i$, as well as the ratio of distances. Therefore, executing the construction of Lemma~\ref{lem:line-coreset} on the points $\Pi(Y_i)$ returns the set $\Pi(Q_i)$.
\end{proof}

\begin{corollary}\label{cor:sensitivity-bound-project}
Let $Y = \{ y_1,\dots, y_n \} \subset \R^d$ be a set of points lying on $k$ lines in $\R^d$. 
\begin{itemize}
\item Let $A_1, \dots, A_s$ denote a partition of $Y$ where each $A_i$ is a $3$-coreset for $(k,\infty)$-line approximation of $Y$ from Lemma~\ref{lem:coreset-for-lines} on the set $Y \setminus \bigcup_{i' = 1}^{i-1} A_{i'}$. 
\item Let $\Pi \colon \R^d \to \R^t$ be any linear map. 
\end{itemize}
For any set of $k$ lines $L' = \{ \ell_1' ,\dots, \ell_k' \}$ in $\R^t$, if $y \in A_i$, we have
\[ \| \Pi(y) - \rho_{L'}(\Pi(y)) \|_2 \leq \cost_z^z(\Pi(Y), L') \cdot \dfrac{3}{i}. \]
\end{corollary}

\begin{proof}
The proof follows from applying the same observation of  Lemma~\ref{lem:lines-sensitivity} to $\Pi(A_j)$, which is a $3$-coreset for $(k, \infty)$-line approximation by Lemma~\ref{lem:coreset-for-lines}. Namely, for every $j \leq i$, the set $\Pi(A_j)$ is a $3$-coreset for $(k,\infty)$-line approximation of a set containing $y$. Thus, if $\| y - \rho_{L'}(y) \|_2 = r$, there must be a set of at least $i$ points $y' \in Y$ where $\|y' - \rho_{L'}(y')\|_2 \geq r / 3$.
\end{proof}

\subsection{Improving the approximation}

We now instantiate the sensitivity function of Lemma~\ref{lem:coreset-for-lines}, and use Corollary~\ref{cor:bad-approx-lines} and Lemma~\ref{lem:lines-events-to-approx} in order to improve on the approximation. Similarly to before, we show that event $\bE_3(1+\eps)$ occurs with sufficiently high probability over the draw of $\bPi$ and $(\bS, \bw)$ by giving an upper bound on the variance as in Lemma~\ref{lem:variance-bound}, Lemma~\ref{lem:variance-bound-sa}, and Lemma~\ref{lem:variance-bound-flat}.

Fix $X = \{ x_1,\dots, x_n \} \subset \R^d$, and let $L = \{ \ell_1,\dots, \ell_k \}$ be the optimal $(k,z)$-line approximation of $X$ in $\R^d$. For $x \in X$, we let $y_x \in \R^d$ be given by $y_x = \rho_{\ell(x)}(x)$, and we denote the set $Y = \{ y_x : x \in X\}$. The sensitivity function $\sigma \colon X \to \R_{\geq 0}$ is specified by Lemma~\ref{lem:lines-sensitivity}. Recall that we first let $A_1,\dots, A_s$ denote a partition $Y$, where $A_i$ is the $3$-coreset for $(k,\infty)$-line approximation of $Y$ from Lemma~\ref{lem:coreset-for-lines}. For $x \in X$ with $y_x \in A_i$, we have
\[ \sigma(x) = 2^{z-1} \cdot \dfrac{\| x - y_x\|_2^z}{\cost_z^z(X, L)} + 2^{2z-1} \cdot \frac{3}{i}. \]
We let $\bE_4$ denote the following event with respect to the randomness in $\bPi \sim \calJ_{d,t}$. For each $x \in X$, we let $\bD_x \in \R_{\geq 0}$ denote the random variable
\[ \bD_x = \dfrac{\| \bPi(x) - \bPi\left( \rho_{\ell(x)}(x) \right)\|_2}{\| x - \rho_{\ell(x)}(x) \|_2},\]
and as in previous sections, event $\bE_4$, which occurs with probability at least $0.99$, whenever 
\[ \sum_{x \in x} \bD_x^{2z} \cdot \sigma(x) \leq 100 \cdot 2^z \cdot \frakS_{\sigma}. \]

\begin{lemma}\label{lem:improving-approx-lines}
Let $\Pi \in \calJ_{d,t}$ be a Johnson-Lindenstrauss map where, for $\alpha > 1$, the following events hold:
\begin{itemize}
\item\label{en:ahah1} Guarantee from Lemma~\ref{lem:upper-bound-lines}: $\sum_{x \in X} \| \Pi(x) - \Pi(\rho_{\ell(x)}(x))\|_2^z \leq \alpha \cdot \cost_z^z(X, L)$.
\item\label{en:ahah2} Guarantee from Corollary~\ref{cor:bad-approx-lines}: letting $L' = \{ \ell_1',\dots, \ell_k' \}$ be the optimal $(k,z)$-line approximation fo $\Pi(X)$, then $\cost_z^z(X, L) \leq \alpha \cdot \cost_z^z(\Pi(X), L')$.
\item Event $\bE_4$ holds. 
\end{itemize}
Then, if we let $(\bS, \bw)$ denote $m = \poly((\log n)^k, 1/\eps, \alpha)$, i.i.d draws from $\tilde{\sigma}$ and $\bw(x) = 1 / (m\tilde{\sigma}(x))$, with probability at least $0.99$, 
\[ \cost_z^z(( \Pi(\bS), \bw), L') \leq (1+\eps) \cdot \cost_z^z(\Pi(X), L'). \]
\end{lemma}

\begin{proof}
We bound the variance,
\begin{align}
\Var_{\bS}\left[ \Ex_{\bx \sim \bS}\left[\dfrac{1}{\tilde{\sigma}(\bx)} \cdot \dfrac{\|\Pi(\bx) - \rho_{L'}(\Pi(\bx))\|_2^z}{\cost_z^z(\Pi(X), L')}  \right]\right] \leq \dfrac{\frakS_{\sigma}}{m} \sum_{x \in X} \left(\dfrac{1}{\sigma(x)} \cdot \dfrac{\|\Pi(x) - \rho_{L'}(\Pi(x))\|_2^{2z}}{\cost_z^{2z}(\Pi(X), L') }\right).\label{eq:var-bound}
\end{align}
We note that, as before, we will apply H\"{o}lder's inequality and the triangle inequality, followed by Corollary~\ref{cor:sensitivity-bound-project}. Specifically, suppose $x \in X$ with $y_x \in A_i$, then,
\begin{align*}
\| \Pi(x) - \rho_{L'}(\Pi(x))\|_2^z &\leq 2^{z-1} \cdot \bD_x^z \cdot \| x - \rho_{L}(x)\|_2^z + 2^{z-1} \cdot \| \Pi(y_x) - \rho_{L'}(\Pi(y_x)) \|_2^z \\
		&\leq 2^{z-1}  \cdot \bD_x^z \cdot \| x - \rho_{L}(x)\|_2^z + 2^{z-1} \cdot \cost_z^z(\Pi(Y), L')\cdot \frac{3}{i} \\
		&\leq 2^{z-1} \cdot \bD_x^z \cdot \|x - \rho_L(x)\|_2^z + \dfrac{2^{2z-2} \cdot 3}{i} \left( \cost_z^z(\Pi(X), L') + \cost_z^z(\Pi(X), \Pi(Y)) \right) .
\end{align*}
We note that from (\ref{en:ahah1}) and (\ref{en:ahah2}), we have $\cost_z^z(\Pi(X), \Pi(Y)) \leq \alpha \cost_z^z(X, L) \leq \alpha^2 \cost_z^z(\Pi(X), L')$. So the above simplifies to
\begin{align*}
\dfrac{\| \Pi(x) - \rho_{L'}(\Pi(x))\|_2^z}{\cost_z^z(\Pi(X), L')} &\leq 2^{z-1} \cdot \bD_x^z \cdot \dfrac{\|x - \rho_L(x)\|_2^z}{\cost_z^z(X, L)} + 2^{2z-2} \cdot \dfrac{3(1+\alpha^2)}{i} \\
	&\leq (\bD_x^z + 1 + \alpha^2) \cdot \sigma(x).
\end{align*}
We now continue upper bounding (\ref{eq:var-bound}), where the variance becomes less than
\begin{align*}
\dfrac{\frakS_{\sigma}}{m} \sum_{x \in X} \left( \bD_x^z + 1 + \alpha^2\right)^2 \cdot \sigma(x) \lsim \dfrac{\frakS_{\sigma}^2 \cdot \alpha^4}{m},
\end{align*}
since event $\bE_4$ holds. Since $\frakS_{\sigma} \leq \poly(2^{2z}, (\log n)^k)$, we obtain our desired bound on the variance by letting $m$ be a large enough polynomial of $(\log n)^k$, $\alpha$, and $1/\eps$.
\end{proof}

\begin{corollary}
Let $X = \{ x_1,\dots, x_n \} \subset \R^d$ be any set of points, and let $L = \{ \ell_1,\dots, \ell_k \}$ be the optimal set of $k$ lines for $(k, z)$-line approximation of $X$. For any $\eps \in (0, 1/2)$, let $\calJ_{d,t}$ be the Johnson-Lindenstrauss map with 
\[ t \gsim \dfrac{k \log \log n + z + \log(1/\eps)}{\eps^3}. \]
Then, with probability at least $0.92$ over the draw of $\bPi \sim \calJ_{d,t}$,
\[ \min_{\substack{L' \text{ $k$ lines} \\ \text{in $\R^t$}}} \cost_z(\bPi(X), L') \geq \dfrac{1}{(1+\eps)^{3 + 1/z}} \cdot \cost_z(X, L). \]
\end{corollary}


\appendix

\section{On preserving ``all solutions'' and comparisons to prior work}\label{sec:forall}

\newcommand\ddfrac[2]{\frac{\displaystyle #1}{\displaystyle #2}}

This section is meant for two things:
\begin{enumerate}
\item To help compare the guarantees of this work to that of prior works on $(k,z)$-clustering of \cite{MMR19} and $(k,2)$-subspace approximation \cite{CEMMP15}, expanding on the discussion in the introduction. In short, for $(k,z)$-clustering, the results of \cite{MMR19} are qualitatively stronger than the results obtained here. In $(k,2)$-subspace approximation, the ``for all'' guarantees  of \cite{CEMMP15} are for the qualitatively different problem of low-rank approximation. While the costs of low-rank approximation and $(k, 2)$-subspace approximation happen to agree at the optimum, the notion of a candidate solution is different.
\item To show that, for two related problems of ``medoid'' and ``column subset selection,'' one cannot apply the Johnson-Lindenstrauss transform to dimension $o(\log n)$ while preserving the cost. The medoid problem is a center-based clustering problem, and column subset selection problem is a subspace approximation problem. The instances we will construct for these problems are very symmetric, such that uniform sampling will give small coresets. These give concrete examples ruling out a theorem which directly relates the size of coresets to the effect of the Johnson-Lindenstrauss transform.
\end{enumerate}

\paragraph{Center-Based Clustering} Consider the following (slight) modification to the center-based clustering problems known as the ``medoid'' problem.
\begin{definition}[$1$-medoid problem]
Let $X = \{ x_1,\dots, x_n \} \subset \R^d$ be any set of points. The $1$-medoid problem asks to optimize
\[ \min_{\substack{c \in X}} \sum_{x \in X} \| x - c\|_2^2. \]
\end{definition}
Notice the difference between $1$-medoid and $1$-mean: in $1$-medoid the center is restricted to be from within the set of points $X$, whereas in $1$-mean the center is arbitrary. Perhaps surprisingly, this modification has a dramatic effect on dimension reduction.
\begin{theorem}\label{thm:medoid}
For large enough $n, d \in \N$, there exists a set of points $X \subset \R^d$ (in particular, given by the $n$-basis vector $\{ e_1,\dots, e_n \} \subset \R^n$) such that, with high probability over the draw of $\bPi \sim \calJ_{d, t}$ where $t = o(\log n)$, 
\begin{align*}
\ddfrac{\mathop{\min}_{c \in X} \sum_{x \in X} \| x - c\|_2^2} {\mathop{\min}_{c' \in \bPi(X)} \sum_{x \in X} \| \bPi(x) - c'\|_2^2}  \geq 2 - o(1).
\end{align*}
\end{theorem}

Theorem~\ref{thm:medoid}  gives very strong lower bound for dimension reduction for $k$-medoid, showing that decreasing the dimension to any $o(\log n)$ does not preserve (even the optimal) solutions within better-than factor $2$. This is in stark contrast to the results on center-based clustering, where the $1$-mean problem can preserve the solutions up to $(1\pm \eps)$-approximation without any dependence on $n$ or $d$. The proof itself is also very straight-forward: each $\bPi(e_i)$ is an independent Gaussian vector in $\R^t$, and if $t = o(\log n)$, with high probability, there exists an index $i \in [n]$ where $\| \bPi(e_i) \|_2^2 = o(1)$. In a similar vein, with high probability $\sum_{i=1}^n \| \bPi(e_i) \|_2^2 \leq (1+o(1)) n$. We take a union bound and set the center $c' = \bPi(e_i)$ for the index $i$ where $\| \bPi(e_i) \|_2^2 = o(1)$. By the pythagorean theorem, the cost of this $1$-medoid solution is at most $(1+o(1)) n$. On the other hand, every $1$-medoid solution in $X$ has cost $2(n-1)$. 

We emphasize that Theorem~\ref{thm:medoid} does not contradict~\cite{MMR19, BBCGS19}, even though it rules out that ``all candidate \emph{centers}'' are preserved. The reason is that the notion of ``candidate solution'' is different. Informally, \cite{MMR19} shows that for any dataset $X \subset \R^d$ of $n$ vectors and any $k \in \N$, $\eps > 0$, applying the Johnson-Lindenstrauss map $\bPi \sim \calJ_{d, t}$ with $t = O(\log(k/\eps) / \eps^2)$ satisfies the following guarantee: for all partitions of $X$ into $k$ sets, $(P_1, P_2, \dots, P_k)$, the following is true:
\begin{align*}
\sum_{\ell=1}^k \min_{c_{\ell}' \in \R^t} \sum_{x \in P_{\ell}} \| \bPi(x) - c_{\ell}'\|_2^2 \approx_{1\pm \eps} \sum_{\ell=1}^k \min_{c_{\ell} \in \R^d} \sum_{x \in P_{\ell}} \| x - c_{\ell}\|_2^2 .
\end{align*}
The ``for all'' quantifies over clusterings $(P_1,\dots, P_k)$ is different (as seen from the $1$-medoid example) from the ``for all'' over centers $c_1,\dots, c_k$.

\paragraph{Subspace Approximation} The same subtlety appears in subspace approximation. Here, we can consider the $1$-column subset selection problem, which at a high level, is the medoid version of subspace approximation. We want to approximate a set of points by their projections onto the subspace spanned by one of those points.
\begin{definition}[$1$-column subset selection]
Let $X = \{ x_1,\dots, x_n \} \subset \R^d$ be any set of points. The $1$-column subset selection problem asks to optimize
\begin{align*}
\min_{\substack{S = \mathrm{span}(\{ x_i \}) \\ x_i \in X}} \sum_{x \in X} \| x - \rho_{S}(x) \|_2^2
\end{align*}
\end{definition}

Again, notice the difference between $1$-column subset selection and $(k, 1)$-subspace approximation: the subspace $S$ is restricted to be in the span of a point from $X$. Given Theorem~\ref{thm:medoid}, it is not surprising that Johnson-Lindenstrauss cannot reduce the dimension of $1$-column subset selection to $o(\log n)$ without incurring high distortions.
\begin{theorem}\label{thm:1-css}
For large enough $n, d \in \N$, there exists a set of points $X \subset \R^d$ such that, with high probability over the draw of $\bPi \sim \calJ_{d,t}$ where $t = o(\log n)$, 
\begin{align*}
\ddfrac{\mathop{\min}_{\substack{S = \mathrm{span}(x) \\ x \in X}} \sum_{x \in X} \| x - \rho_S(x)\|_2^2} {\mathop{\min}_{\substack{S' = \mathrm{span}(\bPi(x)) \\ x \in X}} \sum_{x \in X} \| \bPi(x) - \rho_{S'}(\bPi(x))\|_2^2}  \geq 3/2 - o(1).
\end{align*}
\end{theorem}

The proof is slightly more involved. The instance sets $d = n + 1$, and sets $X = \{ x_1,\dots, x_n \}$ where $x_i = (e_{n+1} + e_i) / \sqrt{2}$. For any subspace $S$ spanned by any of the points $x_i$, via a straight-forward calculation, the distance between $x_j$ and $\rho_S(x_j)$ is $\sqrt{3/4}$ when $j \neq i$, and therefore, the cost of $1$-column subset selection in $X$ is $3/4 \cdot (n-1)$. We apply dimension reduction to $t = o(\log n)$ and we think of $\bg_1,\dots, \bg_{n+1} \in \R^t$ as the independent Gaussian vectors given by $\bPi(e_1),\dots, \bPi(e_{n+1})$. As in the $1$-medoid case, there exists an index $i \in [n]$ for which $\| \bg_i \|_2^2 = o(1)$, and notice that when this occurs, $\bPi(x_i)$ is essentially $\bg_{n+1} / \sqrt{2}$ (because $\|\bPi(x_i) - \bg_{n+1} / \sqrt{2}\|_2 = o(1)$). Letting $S$ be the subspace spanned by $\bPi(x_i)$, we get that the distance between the projection $\| \bPi(x_j) - \rho_{S}(\bPi(x_j)) \|_2^2$ is at most $\| \bg_j \|_2^2 / 2 + o(1)$. This latter fact is because the subspace spanned by $S$ is essentially spanned by $\bg_{n+1}$. Therefore, the cost of the $1$-column subset selection of $\bPi(X)$ is at most $n/2 (1+o(1))$. 

As above, Theorem~\ref{thm:1-css} does not contradict \cite{CEMMP15}, even though it means that ``all candidate \emph{subspaces}'' are preserved needs to be carefully considered. The notion of ``candidate solutions'' is different. In the matrix notation that \cite{CEMMP15} uses, the points in $X$ are stacked into rows of an $n \times d$ matrix (which we denote $X$). A Johnson-Lindenstrauss map $\bPi$ is represented by a $d \times t$ matrix, and applying the map to every point in $X$ corresponds to the operation $X \bPi$ (which is now an $n \times t$ matrix). \cite{CEMMP15} shows that if $\bPi$ is sampled with $t = O(k/\eps^2)$, the following occurs with high probability. For all rank-$k$ projection matrices $P \in \R^{n \times n}$, we have
\[ \| X - PX\|_F^2 \approx_{1\pm \eps} \| X\bPi - P X \bPi \|_F^2. \]
Note that when we multiply the matrix $X$ on the left-hand side by $P$, we are projecting the $d$ columns of $X$ to a $k$-dimensional subspace of $\R^n$. This is different from approximating all points in $X$ with a $k$-dimensional subspace in $\R^d$, which would correspond to finding a rank-$k$ projection matrix $S \in \R^{d \times d}$ and considering $\| X - X S\|_F^2$. In the matrix notation of \cite{CEMMP15}, the dimension reduction result for $(k, 2)$-subspace approximation says that
\begin{align} 
\min_{\substack{S \in \R^{d\times d} \\ \text{ rank-$k$} \\ \text{projection}}} \| X - XS \|_F^2 \approx_{1\pm \eps} \min_{\substack{S' \in \R^{t\times t} \\ \text{ rank-$k$} \\ \text{projection}}} \| X\bPi - X \bPi S' \|_F^2. \label{eq:subs-apprx}
\end{align}
At the optimal $S \in \R^{d \times d}$ and the optimal $P \in \R^{n \times n}$, the costs coincide (a property which holds only for $z = 2$). Thus, \cite{CEMMP15} implies (\ref{eq:subs-apprx}), but it does not say that the cost of all subspaces of $\R^d$ are preserved (as there is a type mismatch in the rank-$k$ projections on the left- and right-hand side of (\ref{eq:subs-apprx})).

\bibliographystyle{alpha}
\bibliography{waingarten}

\end{document}